\documentclass[a4paper,11pt]{article}
\pdfoutput=1
\usepackage[T1]{fontenc}
\usepackage[utf8]{inputenc}
\usepackage{fullpage}

\usepackage[all=normal,bibliography=tight]{savetrees}

\usepackage{amsmath,amssymb,amsfonts,amsthm}
\usepackage[vlined, ruled, linesnumbered]{algorithm2e}
\usepackage[tight]{subfigure}
\usepackage{tikz}
\usepackage{todonotes}
\usepackage{enumerate}
\usepackage{paralist}
\usepackage[absolute]{textpos}

\definecolor{darkblue}{rgb}{0,0,0.45}
\definecolor{darkred}{rgb}{0.6,0,0}
\definecolor{darkgreen}{rgb}{0.13,0.5,0}
\usepackage{hyperref}
\hypersetup{colorlinks, linkcolor=darkblue, citecolor=darkgreen,urlcolor=darkblue}
\usepackage{xcolor}

\usetikzlibrary{shapes,snakes}
\newtheorem{theorem}{Theorem}

\newtheorem{lemma}[theorem]{Lemma}
\newtheorem{claim}[theorem]{Claim}
\newenvironment{claimproof}[0]{\noindent{}\emph{Proof of claim.}}{\hfill $\diamondsuit$\medskip}

\usepackage{etoolbox}

\newcommand{\Oh}[0]{{\mathcal{O}}}

\newcommand{\Ll}[0]{\mathcal{L}}
\newcommand{\Kk}[0]{\mathcal{K}}
\newcommand{\Pp}[0]{\mathcal{P}}
\newcommand{\Dd}[0]{\mathcal{D}}
\newcommand{\Qq}[0]{\mathcal{Q}}

\newcommand{\speq}[0]{\supseteq}
\DeclareMathOperator{\fvs}{fvs}
\DeclareMathOperator{\cp}{cp}
\DeclareMathOperator{\dtw}{dtw}

\newcommand{\df}[0]{\ensuremath{:=}}

\makeatletter
\newcommand\footnoteref[1]{\protected@xdef\@thefnmark{\ref{#1}}\@footnotemark}
\makeatother

\begin{document}
\title{Packing Directed Cycles Quarter- and Half-Integrally\footnote{
    An extended abstract of this manuscript appeared at \emph{European Symposium on Algorithms 2019}~\cite{ESAep}.
}}

\author{Tom\'{a}\v{s} Masa\v{r}\'{i}k\footnote{Department of Applied Mathematics, Charles University, Czech Republic \& Faculty of Mathematics, Informatics and Mechanics, University of Warsaw, Poland, \texttt{masarik@kam.mff.cuni.cz}} 
  \and Irene Muzi\footnote{Technische Universit\"{a}t Berlin, Germany, \texttt{irene.muzi@tu-berlin.de}} 
    \and Marcin Pilipczuk\footnote{Faculty of Mathematics, Informatics and Mechanics, University of Warsaw, Poland, \texttt{malcin@mimuw.edu.pl}}
    \and Pawe\l{} Rz\k{a}\.{z}ewski\footnote{Faculty of Mathematics and Information Science, Warsaw University of Technology, Warsaw, Poland, \texttt{p.rzazewski@mini.pw.edu.pl}}
    \and Manuel Sorge\footnote{Faculty of Mathematics, Informatics and Mechanics, University of Warsaw, Poland, \texttt{manuel.sorge@mimuw.edu.pl}}
}
\date{}

\maketitle

\begin{abstract}
The celebrated Erd\H{o}s-P\'{o}sa theorem states that every undirected graph that does not admit 
a family of $k$ vertex-disjoint cycles contains a feedback vertex set (a set of vertices hitting all cycles in the graph)
of size $\Oh(k \log k)$. 
After being known for long as Younger's conjecture, a similar statement for directed graphs has
been proven in 1996 by Reed, Robertson, Seymour, and Thomas.
However, in their proof, the dependency of the size of the feedback vertex set on the size of vertex-disjoint cycle packing
is not elementary. 

We show that if we compare the size of a minimum feedback vertex set in a directed graph with \emph{quarter-integral} cycle packing
number, we obtain a polynomial bound. More precisely, we show that if in a directed graph $G$ there is no family of
$k$ cycles such that every vertex of $G$ is in at most \emph{four} of the cycles, then there exists a
feedback vertex set in $G$ of size $\Oh(k^4)$. 
Furthermore, a variant of our proof shows that 
if in a directed graph $G$ there is no family of $k$ cycles such that every vertex of $G$ is in at most \emph{two} of the cycles, then there exists a
feedback vertex set in $G$ of size $\Oh(k^6)$. 

On the way there we prove a more general result about quarter-integral packing of subgraphs of high directed treewidth:
for every pair of positive integers $a$ and $b$, if a directed graph $G$ has directed treewidth $\Omega(a^6 b^{8} \log^2(ab))$, 
then one can find in $G$ a family of $a$ subgraphs, each of directed treewidth at least $b$, such that
every vertex of $G$ is in at most four subgraphs.
\end{abstract}

\section{Introduction}

The theory of graph minors, developed over the span of over 20 years by Robertson and Seymour, had a tremendous impact on 
the area of graph algorithms. 
Arguably, one of the cornerstone contributions is the notion of \emph{treewidth}~\cite{RobertsonS84} and the deep understanding
of obstacles to small treewidth, primarily in the form of the \emph{excluded grid theorem}~\cite{grid-minor-poly,RobertsonS86,RobertsonST94}.

Very tight relations of treewidth and the size of the largest grid as a minor in sparse graph classes, such
as planar graphs or graphs excluding a fixed graph as a minor, led to the rich and fruitful theory of bidimensionality~\cite{DemaineH08}.
In general graphs, fine understanding of the existence of well-behaved highly-connected structures (not necessarily grids)
in graphs of high treewidth has been crucial to the development of efficient approximation algorithms
for the \textsc{Disjoint Paths} problem~\cite{ChuzhoyL16}.

In undirected graphs, one of the first theorems that gave some well-behaved structure in a graph 
that is in some sense highly connected is the famous Erd\H{o}s-P\'{o}sa theorem~\cite{ErdosP65} linking the feedback vertex set number
of a graph (the minimum number of vertices one needs to delete to obtain an acyclic graph) and the cycle packing number
(the maximum possible size of a family of vertex-disjoint cycles in a graph). 
The Erd\H{o}s-P\'{o}sa theorem states that a graph that does not contain a family of $k$ vertex-disjoint cycles
has feedback vertex set number bounded by $\Oh(k \log k)$.

A similar statement for directed graphs, asserting that a directed graph without a family of $k$ vertex-disjoint cycles
has feedback vertex set number at most $f(k)$, has been long known as the Younger's conjecture until finally proven
by Reed, Robertson, Seymour, and Thomas in 1996~\cite{ReedRST96}. 
However, the function $f$ obtained in~\cite{ReedRST96} is not elementary; in particular, the proof relies on the Ramsey theorem
for $\Theta(k)$-regular hypergraphs. 
This is in contrast with the (tight) $\Theta(k \log k)$ bound in undirected graphs. 

Our main result is that if one compares the feedback vertex set number of a directed graph
to the \emph{quarter-integral} and \emph{half-integral} cycle packing number (i.e., the maximum size of a family of cycles in $G$
such that every vertex lies on at most four resp. two cycles), one obtains a polynomial bound.

\begin{theorem}\label{thm:ep-all}
Let $G$ be a directed graph that does not contain a family of $k$ cycles such that every vertex in $G$ is contained
in at most $p$ cycles. 
\begin{compactenum}[a)]
\item If $p=4$, then there exists a feedback vertex set in $G$ of size $\Oh(k^4)$,
\item If $p=3$, then there exists a feedback vertex set in $G$ of size $\Oh(k^5)$,
\item If $p=2$, then there exists a feedback vertex set in $G$ of size $\Oh(k^6)$.
\end{compactenum}
\end{theorem}

We remark that if one relaxes the condition even further to a \emph{fractional cycle packing},%
\footnote{A \emph{fractional cycle packing} assigns to every cycle $C$ in $G$ a non-negative real weight $w(C)$ such that for every $v \in V(G)$ the total weight of all cycles containing $v$ is at most $1$. The \emph{size} of a fractional cycle packing is the sum of the weights of all cycles in the graph.}
Seymour~\cite{Seymour95} proved that a directed graph without a fractional cycle packing of size at least
$k$ admits a feedback vertex set of size $\Oh(k \log k \log \log k)$.

\medskip

\emph{Directed treewidth} is a directed analog of the successful notion of treewidth, introduced in~\cite{JohnsonRST01,Reed99}.
An analog of the excluded grid theorem for directed graphs has been conjectured by Johnson, Roberston, Seymour,
and Thomas~\cite{JohnsonRST01} in 2001 and finally proven by Kawarabayashi and Kreutzer in 2015~\cite{DBLP:conf/stoc/KawarabayashiK15}.
Similarly as in the case of the directed Erd\H{o}s-P\'{o}sa property, the relation between the directed treewidth
of a graph and a largest directed grid as a minor in~\cite{DBLP:conf/stoc/KawarabayashiK15} is not elementary.

For a directed graph $G$, let $\fvs(G)$, $\dtw(G)$, and $\cp(G)$ denote the feedback vertex set number,
directed treewidth, and the cycle packing number of $G$, respectively. The following lemma is a restatement of the result of Amiri, Kawarabayashi, Kreutzer, and Wollan~\cite[Lemma 4.2]{AmiriKKW16}: 

\begin{lemma}[{\cite[Lemma 4.2]{AmiriKKW16}}]
\label{lem:AKKW}
Let $G$ be a directed graph with $\dtw(G) \leq w$. For each strongly connected directed graph $H$, the graph $G$ has either $k$ disjoint copies of $H$ as a topological minor, or contains a set $T$ of at most $k \cdot (w+1)$ vertices such that $H$ is not a topological minor of $G-T$. 
\end{lemma}
\noindent Note that the authors of~\cite{AmiriKKW16} prove Lemma~\ref{lem:AKKW} for both topological and butterfly minors, but the previous restatement is sufficient for our purposes.

By taking $H$ as the directed 2-cycle it is easy to derive the following bound:

\begin{lemma}\label{lem:dtw-cp-fvs}
For a directed graph $G$ it holds that 
\[
  \fvs(G)\le(\dtw(G)+1)(\cp(G)+1).
\]
\end{lemma}

In the light of Lemma~\ref{lem:dtw-cp-fvs} and 
since  a directed grid minor of size $k$ contains $k$ vertex-disjoint cycles, the directed grid theorem
of Kawarabayashi and Kreutzer~\cite{DBLP:conf/stoc/KawarabayashiK15} is a generalization of the directed Erd\H{o}s-P\'{o}sa property
due to Reed, Robertson, Seymour, and Thomas~\cite{ReedRST96}.

Theorem~\ref{thm:ep-all} is a direct corollary of Lemma~\ref{lem:dtw-cp-fvs} and the following statement that we prove.
\begin{theorem}\label{thm:dtw-ep-all}
Let $G$ be a directed graph that does not contain a family of $k$ cycles such that every vertex in $G$ is contained
in at most $p$ cycles. 
\begin{compactenum}[a)]
\item If $p=4$, then $\dtw(G) = \Oh(k^3)$,
\item If $p=3$, then $\dtw(G) = \Oh(k^4)$,
\item If $p=2$, then $\dtw(G) = \Oh(k^5)$.
\end{compactenum}
\end{theorem}

\noindent Furthermore, if one asks not for a cycle packing, but a packing of subgraphs of large directed treewidth, 
  we prove the following packing result.
 \begin{theorem}\label{thm:qp}
 There exists an absolute constant $c$ with the following property.
 For every pair of positive integers $a$ and $b$, and every directed graph $G$
 of directed treewidth at least $c\cdot a^6 \cdot b^8 \cdot \log^2(ab)$, there are directed graphs $G_1,G_2,\ldots,G_a$ with the following properties:
 \begin{compactenum}
 \item each $G_i$ is a subgraph of $G$,
 \item each vertex of $G$ belongs to at most four graphs $G_i$, and
 \item each graph $G_i$ has directed treewidth at least $b$.
 \end{compactenum}
 \end{theorem}
\noindent Note that by setting $b=2$ in Theorem~\ref{thm:qp}, one obtains the case $p=4$ of Theorem~\ref{thm:dtw-ep-all} with a slightly weaker bound of $\Oh(k^6 \log^2 k)$ and,
     consequently, case $p=4$ of Theorem~\ref{thm:ep-all} with a weaker bound of $\Oh(k^7 \log^2 k)$.

Theorem~\ref{thm:qp} should be compared to its undirected analog of Chekuri and Chuzhoy~\cite{ChekuriC13} that asserts
that in an undirected graph $G$ of treewidth at least $c \min (ab^2, a^3b)$ one can find $a$ vertex-disjoint subgraphs
of treewidth at least $b$. While we still obtain a polynomial bound, we can only prove the existence of a quarter-integral
(as opposed to integral, i.e., vertex-disjoint) packing of subgraphs of high directed treewidth. 

In the \textsc{Disjoint Paths} problem,
given a graph $G$ and a set of terminal pairs $(s_i,t_i)_{i=1}^k$, we
ask to find an as large as possible collection of vertex-disjoint paths such that every path in the collection
connects some $s_i$ with $t_i$.
Let $\mathrm{OPT}$ be the number of paths in the optimum solution; we say that a family $\Pp$ is a \emph{congestion-$c$ polylogarithmic approximation}
if every path in $\Pp$ connects a distinct pair $(s_i,t_i)$, each vertex of $V(G)$ is contained in at most $c$ paths of $\Pp$, and $|\Pp| \geq \mathrm{OPT} / \mathrm{polylog}(\mathrm{OPT})$.
The successful line of research of approximation algorithms for the \textsc{Disjoint Paths} problem in undirected graphs
leading in particular to a congestion-2 polylogarithmic approximation algorithm of Chuzhoy and Li~\cite{ChuzhoyL16} for the edge-disjoint version, would not be possible without a
 fine understanding of well-behaved well-connected structures in a graph of high treewidth.
Of central importance to such \emph{routing} algorithms is the notion of a \emph{crossbar}: a crossbar of order $k$ and congestion
$c$ is a subgraph $C$ of $G$ with an \emph{interface} $I \subseteq V(C)$ of size $k$ such that for every matching $M$
on $I$, one can connect the endpoints of the matching edges with paths in $C$ such that every vertex is in at most $c$ paths.
Most of the known approximation algorithms for \textsc{Disjoint Paths} find a crossbar $(C,I)$ with a large set of disjoint paths between $I$ and the set of terminals $s_i$ and $t_i$. While one usually does not control how the paths connect the terminals $s_i$ and $t_i$ to interface vertices of $I$, the ability of the crossbar to connect \emph{any} given matching on the interface leads to a solution.

To obtain a polylogarithmic approximation algorithm, one needs the order of the crossbar to be comparable to the number of terminal pairs, which --- by well-known tools such as \emph{well-linked decompositions}~\cite{ChekuriKS05} --- is of the order of the treewidth of the graph. 
At the same time, we usually allow constant congestion (every vertex can appear in a constant number of paths of the solution, instead of just one). 
Thus, the milestone graph-theoretic result  used in approximation algorithms for \textsc{Disjoint Paths} is the existence of a congestion-2 crossbar
of order $k$ in a graph of treewidth $\Omega(k \mathrm{polylog}(k))$. 

While the existence of similar results for the general \textsc{Disjoint Paths} problem in directed graphs
is implausible~\cite{AndrewsCGKT10}, Chekuri, and Ene proposed to study the case of \emph{symmetric demands} where one
asks for a path from $s_i$ to $t_i$ and a path from $t_i$ to $s_i$ for a terminal pair $(s_i,t_i)$. 
First, they provided an analog of the well-linked decomposition for this case~\cite{ChekuriE14}, and then with Pilipczuk~\cite{ChekuriEP18}
showed the existence of an analog of a crossbar and a resulting approximation algorithm for \textsc{Disjoint Paths}
with symmetric demands in planar directed graphs.
Later, this result has been lifted to arbitrary proper minor-closed graph classes~\cite{CSS17}.
However, the general case remains widely open.

As discussed above, for applications in approximation algorithms for \textsc{Disjoint Paths}, it is absolutely essential
to squeeze as much as possible from the bound linking directed treewidth of a graph with the order of the crossbar,
while the final congestion is of secondary importance (but we would like it to be a small constant). 
We think of Theorem~\ref{thm:qp} as a step in this direction: sacrificing integral packings for quarter-integral ones,
we obtain much stronger bounds than the non-elementary bounds of~\cite{ReedRST96}.
Furthermore, such a step seems necessary, as it is hard to imagine a crossbar of order $k$ that would not contain
a constant-congestion (i.e., every vertex might be used in a constant number of cycles) packing of $\Omega(k)$ directed cycles.

On the technical side, the proof of Theorem~\ref{thm:qp} borrows a number of technical tools from the recent
work of Hatzel, Kawarabayashi, and Kreutzer that proved polynomial bounds for the directed grid minor theorem in planar
graphs~\cite{DBLP:conf/soda/HatzelKK19}.
We follow their general approach to obtain a directed treewidth sparsifier~\cite[Section 5]{DBLP:conf/soda/HatzelKK19} and modify
it in a number of places for our goal. The main novelty comes in different handling of the case
when two linkages intersect a lot. Here we introduce a new partitioning tool (see Section~\ref{sec:sep}) 
which we use in the crucial moment where we separate subgraphs $G_i$ from each other.

\paragraph{Organization and proof outline.}
After brief preliminaries in Section~\ref{sec:prelim}, we prove Theorem~\ref{thm:qp} in Sections~\ref{sec:sep}--\ref{sec:main}. A brief outline of the proof is as follows.
Assuming that the directed treewidth of the graph~$G$ in the statement Theorem~\ref{thm:qp} is sufficiently large, we use a known result (Lemma~\ref{lem:path-system}) to obtain a sufficiently large set~$\Pp$ of paths whose endpoints are well-linked.
We then distinguish two cases.
In the first case, the intersection graph of the paths in~$\Pp$ is sparse---the \emph{sparse case}.
Then, by the properties of $\Pp$ guaranteed by Lemma~\ref{lem:path-system} we can rather directly construct the required graphs~$G_i$: Intuitively, then there is a subset of $\Pp$ whose paths are sufficiently independent from each other to allow for a small overlap of the constructed graphs.
In the second case, the intersection graph of the paths in $\Pp$ contains a dense subgraph---the \emph{dense case}.
To treat this case, we need a new partitioning tool which allows us to separate the dense intersection subgraph into sufficiently many subgraphs that all remain sufficiently dense.
We can then look at each of these dense subgraphs individually and, using the density, construct the required subgraph~$G_i$ of sufficiently large directed treewidth.

The organization is as follows. Section~\ref{sec:sep} introduces the new partitioning tool, Section~\ref{sec:dense} handles the dense case in the analysis, while Section~\ref{sec:main} handles the sparse case and wraps up the argument. 

In Section~\ref{sec:imp}, we discuss how to modify the arguments of Section~\ref{sec:main} to obtain the improved bounds of Theorem~\ref{thm:dtw-ep-all}.

\section{Preliminaries}\label{sec:prelim}

For brevity, we use $[i] := \{1, 2, \ldots, i\}$, where $i \in \mathbb{N} \setminus \{0\}$.

\subsection{Linkages}

\newcommand{\linkfrom}[1]{\ensuremath{A(#1)}}
\newcommand{\linkto}[1]{\ensuremath{B(#1)}}
\newcommand{\pathfrom}[1]{\ensuremath{\textsf{start}(#1)}}
\newcommand{\pathto}[1]{\ensuremath{\textsf{end}(#1)}}

\newcommand{\backlink}{\textsf{back}}
  
\newcommand{\auxg}{\ensuremath{\textsf{Aux}}}

Let $G=(V(G),E(G))$ be a directed graph and let $A, B$ be subsets of $V(G)$ with $|A| = |B|$. A \emph{linkage} from $A$ to   $B$ in~$G$ is a set~$\Ll$ of~$|A|$ pairwise vertex-disjoint paths in~$G$, each with a starting vertex in $A$ and ending vertex in $B$. The \emph{order} of $\Ll$ is $|\Ll|=|A|$.    
For $X, Y \subseteq V(G)$ and a linkage $\Ll$ from $X$ to $Y$, we denote $\linkfrom{\Ll} := X$ and 
$\linkto{\Ll} := Y$.  
For a path or a walk $P$, by $\pathfrom{P}$ and $\pathto{P}$ we denote the starting and ending vertex of $P$, respectively.

Let $\Ll$ and $\Kk$ be linkages. The \emph{intersection graph} of $\Ll$ and $\Kk$, denoted by $I(\Ll, \Kk)$, is the bipartite graph with the vertex set $\Ll \cup \Kk$ and an edge between a vertex in~$\Ll$ and a vertex in~$\Kk$ if the corresponding paths share at least one vertex.

A vertex set~$W \subseteq V(G)$ is \emph{well-linked} if for all subsets   $A, B \subseteq W$ with $|A| = |B|$ there is a linkage $\Ll$ of
  order $|A|$ from $A$ to $B$ in   $G \setminus (W \setminus (A \cup B))$.

Let $\Pp$ be a family of walks in $G$ and let $c$ be a positive integer. We say that $\Pp$ is \emph{of congestion $c$} if for every $v \in V(G)$, the total
number of times the walks in $\Pp$ visit $v$ is at most $c$; here, if a walk $W \in \Pp$ visits $v$ multiple times, we count each visit separately. 
A family of paths $\Pp$ is \emph{half-integral} (\emph{quarter-integral}) if it 
is of congestion $2$ (resp. $4$). 

We call two linkages $\Ll$ and $\Ll^\backlink$
\emph{dual} to each other if $\linkfrom{\Ll} = \linkto{\Ll^\backlink}$ and
$\linkfrom{\Ll^\backlink} = \linkto{\Ll}$.
For two dual linkages $\Ll$ and $\Ll^\backlink$ in a graph $G$, we define 
an auxiliary directed graph $\auxg(\Ll,\Ll^\backlink)$ as follows.
We take $V(\auxg(\Ll,\Ll^\backlink)) = \Ll$ and for every path $P \in \Ll^\backlink$
that starts in a vertex $\pathfrom{P} = \pathto{L}$ for some $L \in \Ll$
and ends in a vertex $\pathto{P} = \pathfrom{L'}$ for some $L' \in \Ll$,
we put an arc $(L,L')$ to $\auxg(\Ll,\Ll^\backlink)$. Note that
it may happen that $L = L'$.
When the backlinkage $\Ll^\backlink$ is clear from the context, we abbreviate
$\auxg(\Ll,\Ll^\backlink)$ to $\auxg(\Ll)$. 
Observe that in $\auxg(\Ll,\Ll^\backlink)$ every node is of in- and out-degree exactly one
and thus this graph is a disjoint union of directed cycles. 

With every arc $(L,L')$ of $\auxg(\Ll, \Ll^\backlink)$ we can associate the walk from
$\pathfrom{L}$ to $\pathfrom{L'}$ that first goes along $L$ and then follows the path
$P \in \Ll^\backlink$ that gives rise to the arc $(L,L')$. 
Consequently, with every collection of pairwise disjoint paths and cycles in $\auxg(\Ll,\Ll^\backlink)$
there is an associated collection of walks (closed walks for cycles) in $G$ that is
of congestion $2$ as it originated from two linkages. Note that the same construction works if $\Ll$ and $\Ll^\backlink$ are half-integral linkages, and then the walks
in $G$ corresponding to a family of paths and cycles in $\auxg(\Ll,\Ll^\backlink)$ would
be of congestion $4$.

Furthermore, with a pair of dual linkages $\Ll$ and $\Ll^\backlink$ we can associate
a \emph{backlinkage-induced order} $\Ll = \{L_1,L_2,\ldots,L_{|\Ll|}\}$ as follows. 
If $C_1,C_2,\ldots,C_r$ are the cycles of $\auxg(\Ll,\Ll^\backlink)$ in an arbitrary order, then
$L_1,L_2,\ldots,L_{|C_1|}$ are the vertices of $C_1$ in the order of their appearance on $C_1$, and
$L_{|C_1|+1},\ldots,L_{|C_1|+|C_2|}$ are the vertices of $C_2$ in the order of their appearance on $C_2$, etc.
That is, we order the elements of $\Ll$ first according to the cycle of $\auxg(\Ll)$ they lie on, and then, within
one cycle, according to the order around this cycle.

We will also need the following operation on a pair of dual linkages $\Ll$ and $\Ll^\backlink$.
Let $\Pp \subseteq \Ll$ be a sublinkage. For every $P \in \Pp$, construct a walk $Q(P)$ as follows.
Start from the path $Q_0 \in \Ll^\backlink$ with $\pathfrom{Q_0} = \pathto{P}$ and set $Q(P) = Q_0$. Given $Q_i \in \Ll^\backlink$ for $i \geq 0$, proceed as follows.
Let $P_{i+1} \in \Ll$ be the path with $\pathto{Q_i} = \pathfrom{P_{i+1}}$. If $P_{i+1} \in \Pp$, then stop. Otherwise, define $Q_{i+1} \in \Ll^\backlink$
to be the path with $\pathto{P_{i+1}} = \pathfrom{Q_{i+1}}$. Append $P_{i+1}$ and $Q_{i+1}$ at the end of $Q(P)$ and repeat. 
Finally, we shortcut $Q(P)$ to a path $Q'(P)$ with the same endpoints.
In this manner, $\Qq := \{Q'(P)~|~P \in \Pp\}$ is a half-integral linkage with $\linkfrom{\Pp} = \linkto{\Qq}$ and $\linkfrom{\Qq} = \linkto{\Pp}$. 
We call $\Qq$ the \emph{backlinkage induced by $\Pp$ on $(\Ll, \Ll^\backlink)$}. 
Furthermore, we can perform the same construction if $\Ll$ and $\Ll^\backlink$ are half-integral linkages, obtaining a quarter-integral linkage $\Qq$.

\subsection{Degeneracy and directed treewidth}

A graph $G$ is \emph{$d$-degenerate} if every subgraph of $G$ contains a vertex of degree at most $d$. 
In this paper we do not need the exact definition of directed treewidth. Instead, we rely on the following two results.
\begin{lemma}[\cite{Reed99}]\label{lem:dtw2wl}
Every directed graph $G$ of directed treewidth $k$ contains a
well-linked set of size $\Omega(k)$.
\end{lemma}

\begin{lemma}[\cite{DBLP:journals/corr/KawarabayashiK14,DBLP:conf/stoc/KawarabayashiK15}]\label{lem:path-system}
There is an absolute constant $c'$ with the following property.
 Let $\alpha,\beta \geq 1$ be integers and let $G$ be a digraph of $\dtw(G) \geq c' \cdot \alpha^2\beta^2$.
  Then there exists a set of $\alpha$ vertex-disjoint paths $P_1,\ldots,P_\alpha$ and sets $A_i,B_i\subseteq V(P_i)$, where $A_i$ appears before $B_i$ on $P_i$, both $|A_i|, |B_i|= \beta$, and
  $\bigcup_{i=1}^{\alpha} A_i\cup B_i$
is well-linked.
\end{lemma}

We also need the following two auxiliary results. 
Note that a coloring in Lemma~\ref{lem:degenerate} can be arbitrary and is not necessarily proper.

\begin{lemma}[{\cite[Lemma~4.3]{DBLP:journals/ejc/ReedW12}}]\label{lem:degenerate}
Let $r\ge 2$, $d$ be a real, and $H$ be an $r$-colored graph with color classes $V_1,\ldots,V_r$, such that for every $i$ it holds that $|V_i|\ge 4e(r-1)d$  and for every $i \neq j$ the graph $H[V_i\cup V_j]$ is $d$-degenerate. Then there exists an independent set $\{x_1,\ldots,x_r\}$ such that $x_i\in V_i$ for every $i \in [r]$.
\end{lemma}

\begin{lemma}[{\cite[Lemma~5.5]{DBLP:conf/soda/HatzelKK19}}]\label{lem:dtw-bound}
  Let $G$ be a digraph and $P_1,\ldots,P_k$ be disjoint paths such that each $P_i$ consists of two subpaths $A_i$ and $B_i$, where $A_i$ precedes $B_i$.
Furthermore, let $\{ L_{i,j} \colon i,j \in [k], i\neq j\}$ be a set of pairwise disjoint paths, such that $L_{i,j}$ starts in $B_i$ and ends in $A_j$.
  Then
  \[
    \dtw\Bigl(\bigcup_i  P_i \cup \bigcup_{i\neq j} L_{i,j}\Bigr)\ge \frac{k}{8}.
  \]
\end{lemma}

\section{Partitioning lemma}\label{sec:sep}
In this section, we develop a main technical tool that we use in the
proof of Theorem~\ref{thm:qp}. Intuitively, in the dense case of the proof (see the proof of Lemma~\ref{lem:dense} in Section~\ref{sec:dense}),
we will have a bipartite graph of large minimum degree which we
partition into subgraphs induced by pairs of vertex sets~$(U_i, W_i)$.
These subgraphs will define the~$G_i$ from the statement of Theorem~\ref{thm:qp}. To
obtain a lower bound on the directed treewidth of~$G_i$, we need that
the parts~$(U_i, W_i)$ each induce a subgraph of large average degree.

The bipartite graph $G=(X \cup Y,E)$, which will be considered in this section, has a fixed ordering of vertices in each bipartition class: $X=\{x_1,x_2,\ldots,x_a\}$ and $Y=\{y_1,y_2,\ldots,y_b\}$.
A subset $X'$ of $X$ (resp. $Y'$ of $Y$) is called a \emph{segment} if it is of the form $\{x_i,x_{i+1},\ldots,x_j\}$ for some $1 \leq i < j \leq a$ (resp. $\{y_i,y_{i+1},\ldots,y_j\}$ for some $1 \leq i < j \leq b$).
Now we are ready to prove the following lemma.

\begin{lemma} \label{lem:partition}
Let $h \geq 0$ and $n$ be integers, $d$ be a positive real such that $d \cdot 4^{h+1} - 1 > 2$, and let $G$ be a bipartite graph with bipartition classes $X = \{x_1,x_2,\ldots,x_a\}$ and $Y=\{y_1,y_2,\ldots,y_b\}$, such that $a+b \leq n$ and $|E(G)| \geq (d \cdot 4^{h+1} -1) \cdot n$.
Then in $X$ we can find $k:=2^h$ pairwise disjoint sets $I_1,I_2,\ldots,I_k$, and in $Y$ we can find $k$ pairwise disjoint sets $J_1,J_2,\ldots,J_k$, such that:
\begin{compactenum}
\item for every $i \in [k]$ the set $I_i$ is a segment of $X$ and the set $J_i$ is a segment of $Y$,
\item for every $i \in [k]$, the number of edges in $G$ between $\{x_j: j \in I_i\}$ and $\{y_j: j \in J_i\}$ is at least $d \cdot n$.
\end{compactenum}
\end{lemma}
\begin{proof}
For $I \subseteq X$ and $J \subseteq Y$, let $e(I,J)$ denote the number of edges with one endpoint in $I$ and the other in $J$. Observe that $e(X,Y)=|E(G)| > 2n$. 

We prove the lemma by induction on $h$. Note that for $h=0$ the claim is trivially satisfied by taking $I_1 = X$ and $J_1 = Y$, as $d \cdot 4^{h+1} - 1 > 2$ and $h \geq 0$ implies $d \cdot 4^{h+1}-1 \geq d$.
So now assume that $h \geq 1$ and the claim holds for $h-1$. Let $s \in [a]$ be the minimum integer, for which $\sum_{i=1}^s \deg x_i \geq e(X,Y)/2$, and let $t \in [b]$ be the minimum integer, for which $\sum_{i=1}^t \deg y_i \geq e(X,Y)/2$.
We observe that $d \cdot 4^{h+1} -1 > 2$ implies that $1 < s < a$ and $1 < t < b$.
Define $X^1 := \{x_1,x_2,\ldots,x_{s-1}\}$ and $X^2 := \{x_{s+1},\ldots,x_a\}$, and $Y^1 := \{y_1,y_2,\ldots,y_{t-1}\}$ and $Y^2 := \{y_{t+1},\ldots,y_b\}$.

We aim to show that the number of edges joining $X^1$ and $Y^1$ is roughly the same as the number of edges joining $X^2$ and $Y^2$, and the number of edges joining $X^1$ and $Y^2$ is roughly the same as the number of edges joining $X^2$ and $Y^1$. Since $\deg x_s \leq b < n$ and $\deg y_t \leq a < n$, by the choice of $s$ and $t$ we obtain the following set of inequalities.
\begin{align} \label{eq:boundXiY}
\begin{split}
e(X,Y)/2 - \deg x_s \leq e(X^1,Y) \leq  e(X,Y)/2\\
e(X,Y)/2 - \deg x_s \leq e(X^2,Y) \leq  e(X,Y)/2\\
e(X,Y)/2 - \deg y_t \leq e(X,Y^1) \leq  e(X,Y)/2\\
e(X,Y)/2 - \deg y_t \leq e(X,Y^2) \leq  e(X,Y)/2.
\end{split}
\end{align}
Observe that
\begin{align*}
e(X^1,Y^1) + e(X^1,Y^2) &\leq e(X^1,Y) = e(X^1,Y^1) + e(X^1,Y^2) + e(X^1,\{y_t\}) \\
                                         &\leq e(X^1,Y^1) + e(X^1,Y^2) + \deg y_t
\end{align*}
(and analogously for each of the remaining inequalities in \eqref{eq:boundXiY}).
Thus we obtain:
\begin{align} \label{eq:bound-sumXY}
\begin{split}
e(X,Y)/2 - n \leq e(X^1,Y^1)+e(X^1,Y^2) \leq  e(X,Y)/2\\
e(X,Y)/2 - n \leq e(X^2,Y^1)+e(X^2,Y^2) \leq  e(X,Y)/2\\
e(X,Y)/2 - n \leq e(X^1,Y^1)+e(X^2,Y^1) \leq  e(X,Y)/2\\
e(X,Y)/2 - n \leq e(X^1,Y^2)+e(X^2,Y^2) \leq  e(X,Y)/2.
\end{split}
\end{align}
By subtracting appropriate pairs of inequalities in \eqref{eq:bound-sumXY}, we obtain the following bounds.
\begin{align} \label{eq:bound-diffXY}
\begin{split}
- n \leq e(X^1,Y^1)-e(X^2,Y^2) \leq  n\\
- n \leq e(X^1,Y^2)-e(X^2,Y^1) \leq  n\\
\end{split}
\end{align}
Recall that 
\begin{align*}
e(X,Y) &=  e(X^1,Y^1) + e(X^1,Y^2) + e(X^2,Y^1) + e(X^2,Y^2) + \deg x_s + \deg y_t \\
    &\leq e(X^1,Y^1) + e(X^1,Y^2) + e(X^2,Y^1) + e(X^2,Y^2) + n.
\end{align*}
Thus, by the pigeonhole principle, at least one of the following holds:
\begin{align} \label{eq:cases}
\begin{split}
e(X^1,Y^1)+e(X^2,Y^2) \geq  e(X,Y)/2 -n/2\\
e(X^1,Y^2)+e(X^2,Y^1) \geq  e(X,Y)/2 -n/2.
\end{split}
\end{align}

Suppose that the first case holds. Define $G^1 := G[X^1 \cup Y^1]$ and $G^2 := G[X^2 \cup Y^2]$. Combining \eqref{eq:bound-diffXY} and \eqref{eq:cases}, we obtain that
\begin{align} 
\begin{split}
|E(G^1)| &=e(X^1,Y^1) \geq  e(X,Y)/4 -3n/4 \geq (d \cdot 4^{h+1}-1)n/4 - 3n/4 \ge (d\cdot 4^h-1)n\\
|E(G^2)| &=e(X^2,Y^2) \geq  e(X,Y)/4 -3n/4 \geq (d\cdot 4^h-1)n.
\end{split}
\end{align}
We observe that graphs $G^1,G^2$ satisfy the inductive assumption (for $h-1$), so in the vertex set of $G^1$ we can find two families of $k/2$ pairwise corresponding segments $I^1_1,I^1_2,\ldots,I^1_{k/2}$ and $J^1_1,J^1_2,\ldots,J^1_{k/2}$, and in the vertex set of $G^2$ we can find two families of $k/2$ pairwise corresponding segments $I^2_1,I^2_2,\ldots,I^2_{k/2}$ and $J^2_1,J^2_2,\ldots,I^2_{k/2}$. We obtain the desired subsegments of $X$ and $Y$ by setting:
\begin{equation*}
\begin{aligned}[c]
I_i =
\begin{cases}
I^1_i & \text{ if } i \leq k/2,\\
I^2_{i-k/2} & \text{ if } i > k/2,
\end{cases}
\end{aligned}
\qquad
\begin{aligned}[c]
J_i =
\begin{cases}
J^1_i & \text{ if } i \leq k/2,\\
J^2_{i-k/2} & \text{ if } i > k/2.
\end{cases}
\end{aligned}
\end{equation*}
If the second case in \eqref{eq:cases} holds, we take $G^1 := G[X^1 \cup Y^2]$ and $G^2 := G[X^2 \cup Y^1]$, and the rest of the proof is analogous.
\end{proof}

The following statement brings the technical statement of
Lemma~\ref{lem:partition} into a more easily applicable form.

\begin{lemma} \label{lem:disjointPairs}
Let $k, r \geq 1$ be two integers and let $G$ be a bipartite graph with bipartition classes $X = \{x_1,x_2,\ldots,x_a\}$ and $Y=\{y_1,y_2,\ldots,y_b\}$ and minimum degree at least $2^9 \cdot r \cdot k$.
Then there are $k$ sets $U_1,U_2,\ldots,U_k$, and $k$ sets $W_1,W_2,\ldots,W_k$, such that:
\begin{compactenum}
\item for each $i \in [k]$ the set $U_i$ is a segment of $X$ and the set $W_i$ is a segment of $Y$,
\item for each distinct $i,j \in [k]$ we have $U_i \cap U_j = \emptyset$ and $W_i \cap W_j = \emptyset$,
\item for every $i \in [k]$, the average degree of the graph $G[U_i \cup W_i]$ is at least $r$.
\end{compactenum}
\end{lemma}

\begin{proof}
Let $h$ be the minimum integer, such that $k' :=2^h \geq 2k$; note that $k' < 4k$. Also, define $d= 2r/k$ and $n=a+b$.
We have
$$d \cdot 4^{h+1} -1 = 4d(k')^2 -1 \geq \frac{8r}{k} \cdot (2k)^2 - 1 = 32 \cdot r \cdot k - 1 > 2.$$
Observe that the number of edges in $G$ is at least 
$$n \cdot r \cdot 2^8\cdot k = (16r/k \cdot (4k)^2)n > (4d (k')^2)n > (d \cdot 4^{h+1} -1)n.$$
Thus $G$ satisfies the assumptions of Lemma \ref{lem:partition} for $h$, $n$, and $d$. Let $I_1,I_2,\ldots,I_{k'}$ be the disjoint segments in $X$, and $J_1,J_2,\ldots,J_{k'}$ be the disjoint segments in $Y$, whose existence is guaranteed by Lemma \ref{lem:partition}. 

A segment $I_i$ ($J_i$, resp.) is called \textit{large} if $|I_i| \geq 2n/k'$ ($|J_i| \geq 2n/k'$, resp.).
A pair $(I_i,J_i)$ is \textit{large} if at least one of $I_i, J_i$ is large, otherwise the pair is \textit{small}.
Note that there are at most $n / (2n/k') = k'/2$ large segments in total.
Thus the number of small pairs is at least $k'/2 \geq k$. We obtain the segments $(U_i,W_i)$ by taking the first $k$ small pairs ($I_i,J_i)$. Clearly these segments satisfy conditions 1.\ and 2.\ of the lemma.

Now take any $i \in [k]$ and let us compute the average degree of the graph $G_i:=G[U_i \cup W_i]$. By Lemma~\ref{lem:partition}, $|E(G_i)| \geq d \cdot n$. On the other hand, since $(U_i,W_i)$ is a small pair, we have that $|V(G_i)| = |U_i \cup W_i| < 4n/k'$. Thus we obtain that the average degree of $G_i$ is
\[
\frac{2\cdot |E(G_i)|}{|V(G_i)|} > \frac{d \cdot n}{4n/k'} = \frac{dk'}{4} \geq d \; \frac{2k}{4} = \frac{2r}{k} \cdot \frac{k}{2} = r.
\]
This completes the proof.
\end{proof}

\section{The dense case}\label{sec:dense}
In this section, we prove Theorem~\ref{thm:qp} roughly in the case
when there are two linkages $\Ll$ and $\Kk$ such that their set
$\linkfrom{\Ll} \cup \linkfrom{\Kk} \cup \linkto{\Ll} \cup
\linkto{\Kk}$ of endpoints is well linked and such that the paths in
$\Ll$ and $\Kk$ intersect a lot. The formal statement proved in this
section is as follows.

\begin{lemma}\label{lem:dense}
  Let $a, b \in \mathbb{N}^+$. Let $D$ be a directed graph and $\Ll$
  and $\Kk$ be two linkages in~$D$ such that
  $\linkfrom{\Ll} \cup \linkto{\Ll} \cup \linkfrom{\Kk} \cup
  \linkto{\Kk}$ is well-linked in~$D$. Suppose that the intersection
  graph $I(\Ll, \Kk)$ has degeneracy more
  than~$327\,680 \cdot a \cdot b \cdot \log_2(|\Ll|/b)$. Then there are
  directed graphs $D_1, D_2, \ldots, D_a$ with the following
  properties:
\begin{compactenum}[(i)]
\item each $D_i$ is a subgraph of~$D$,
\item each vertex of $D$ belongs to at most four graphs~$D_i$, and
\item each graph~$D_i$ has directed treewidth at least~$b$.
\end{compactenum}
\end{lemma}

\paragraph{Proof outline} The basic idea of the proof of
Lemma~\ref{lem:dense} is as follows. We first fix a pair of linkages
$\Ll^\backlink$ and $\Kk^\backlink$ which are dual to $\Ll$ and $\Kk$,
respectively. (This is possible because of well-linkedness of the
endpoints.) The subgraphs $D_i$ that we construct will subpartition
the vertex set of each of the four linkages~$\Ll, \Ll^\backlink, \Kk, \Kk^\backlink$ and
hence each vertex of $G$ is in at most four subgraphs~$D_i$. To
construct the desired subgraphs $D_i$, we consider the
backlinkage-induced order~$\Pi_\Ll$ on $\Ll$
and $\Pi_\Kk$ on $\Kk$.
Using these orderings of the
paths of~$\Ll$ and~$\Kk$, we can apply the partitioning lemma
(Lemma~\ref{lem:disjointPairs}) to the intersection graph of~$\Ll$
and~$\Kk$, obtaining a subpartition $I_1, \ldots, I_k$ of~$\Ll$ and a
subpartition~$J_1, \ldots, J_k$ of~$\Kk$. These subpartitions have the
nice property that each intersection graph~$I(I_i, J_i)$ induced by a
pair $I_i, J_i$ contains many edges (representing intersections
between the corresponding paths) and that only a constant number of
cycles of~$\auxg(\Ll)$ and $\auxg(\Kk)$ cross $I_i$ or $J_i$. By closing
each of these crossing cycles by introducing an artificial new path,
we obtain a pair of dual linkages $I_i, I_i'$, and a pair of dual of
linkages $J_i, J_i'$. Using then
Lemma~\ref{lem:high-degree-linkages-dtw} below, we will obtain a lower
bound on the directed treewidth of the graph induced by $I_i \cup J_i \cup I_i' \cup J_i'$,
      which constitute our desired subgraph $D_i$.

\paragraph{Treewidth lower bound}
For technical reasons, we will have to work with half-integral
linkages. The intersection graph for a pair of
half-integral linkages is defined in the same way as for ordinary
linkages. 
\begin{lemma}\label{lem:high-degree-linkages-dtw}
  Let $k, d \in \mathbb{N}^+$ 
  and
  $\Pp, \Pp^\backlink, \Qq, \Qq^\backlink$ be four half-integral linkages in a
  directed graph such that $\Pp$ and $\Pp^\backlink$ are dual to each other
  and $\Qq$ and $\Qq^\backlink$ are dual to each other. Let the intersection
  graph $I(\Pp, \Qq)$ have minimum degree at least~$d$ where
  $d\ge 8k\log_{4\over 3} {({{|\Pp|}\over 24k})}+24k+4$.
  Then the graph
  $\bigcup(\Pp \cup \Pp^\backlink \cup \Qq \cup \Qq^\backlink)$ has directed treewidth at least~$k$.
\end{lemma}
\noindent The proof of Lemma~\ref{lem:high-degree-linkages-dtw} is
inspired by the proof of Lemma 5.4 in~\cite{DBLP:conf/soda/HatzelKK19}.
We could use Lemma~5.4 here as well, but its proof, unfortunately, contains errors.
Nevertheless, we derive an incomparable bound which is better for our use since the lower bound on the degree that we need depends only linearly on $k$ whereas the lower bound claimed in Lemma~5.4~\cite{DBLP:conf/soda/HatzelKK19} is~$k^2$.
Also, we adapt the constants in the lemma for half-integral linkages.

The proof of Lemma~\ref{lem:high-degree-linkages-dtw} is based on the following
Lemma~\ref{lem:balanced-separation}. Herein, we use the following
definition. Let $D$ be a directed graph. A \emph{separation} in $D$ is
a pair $(X, Y)$ of two vertex subsets~$X, Y \subseteq V(D)$ with
$X \cup Y = V(D)$ such that there are no edges from $Y \setminus X$ to
$X \setminus Y$ in~$D$. The \emph{order} of $(X, Y)$ is $|X \cap Y|$.

\begin{lemma}[\cite{KO14}]\label{lem:balanced-separation}
  Let $w \in \mathbb{N}$. Let $G$ be a directed graph of directed
  treewidth at most $w$ and let $W \subseteq V(G)$ such that
  $|W| \geq 2w + 2$. Then there is a separation $(X, Y)$ in~$G$ of
  order at most~$w$ such that $X$ and $Y$ each contain at least
  $|W|/4$ elements of~$W$.
\end{lemma}
\begin{proof}
  The statement follows easily from Lemma~6.4.10 in~\cite{KO14}.
  We provide a proof for completeness. By Lemma~6.4.10 in~\cite{KO14} there exist three pairwise disjoint vertex sets
  $A, B, S \subseteq V(G)$ such that the following properties hold.
  \begin{itemize}
  \item[(i)] $W = A \cup (S \cap W) \cup B$.
  \item[(ii)] There is no directed path from $B$ to $A$ in $G - S$. 
  \item[(iii)] Both $A$ and $B$ contain at most $3|W|/4$ elements
    of~$W$.
    \item[(iv)] $|S|\leq w$.
  \end{itemize}
  Based on the sets $A, B, S$, we define the desired
  separation~$(X, Y)$. Let $R(B)$ be the set of vertices in
  $V(G) \setminus B$ reachable from $B$, that is, a vertex
  $v \in V(G)$ is in $R(B)$ if it is not in $B$ and there is a
  directed path in $G$ to~$v$ from a vertex in $B$. Note that
  $R(B) \cap A = \emptyset$ by Property~(ii). Define
  $Y = S \cup B \cup R(B)$ and $X = (V(G) \setminus Y) \cup S$. Note
  that $X \cap Y = S$. We claim that $(X, Y)$ is a separation for~$G$
  with the desired properties.

  Clearly, $X \cup Y = V(G)$. Thus, to show that $(X, Y)$ is a
  separation, it remains to show that there is no edge from
  $Y \setminus X$ to $X \setminus Y$. For the sake of contradiction,
  assume that there is such an edge~$(y, x) \in E(G)$ with
  $y \in Y \setminus X$ and $x \in X \setminus Y$. Observe that
  $y \in Y \setminus S = B \cup R(B)$ and thus $x \in B \cup R(B)$.
  Then, $x \in Y$ by definition, a contradiction. Hence, $(X, Y)$ is a
  separation. Recall that $X \cap Y = S$ and thus $(X, Y)$ is of order at
  most~$w$, as required.

  It remains to show the balancedness property. Clearly,
  $B \subseteq Y \setminus X$. Furthermore, since
  $A \cap (S \cup B \cup R(B)) = \emptyset$, we have
  $A \subseteq X \setminus Y$. Thus,
  \begin{align*}
    |W \cap (Y \setminus X)| &= |W \cap B| \leq 3|W|/4 \text{, and}\\
    |W \cap (X \setminus Y)| &= |W \cap A| \leq 3|W|/4.
  \end{align*}
  Hence,
  \begin{align*}
    |W \cap X| &\geq |W| - |W \cap (Y \setminus X)| \geq |W|/4 \text{, and}\\
    |W \cap Y| &\geq |W| - |W \cap (X \setminus Y)| \geq |W|/4.
  \end{align*}
  This completes the proof.
\end{proof}

We are now ready to prove that two pairs of half-integral linkages
whose paths intersect a lot induce a graph with large directed
treewidth.
\begin{proof}[Proof of Lemma~\ref{lem:high-degree-linkages-dtw}]
  Let $D$ be the graph containing~$\Pp, \Pp^\backlink, \Qq$, and $\Qq^\backlink$, and let $H = \bigcup(\Pp \cup \Pp^\backlink \cup \Qq \cup \Qq^\backlink)$. 
  Assume for the sake of contradiction that $H$ has directed treewidth less than $k$.
  The basic idea is to iteratively separate the paths in
  $\Pp$ and $\Qq$ using a balanced separation of small order while
  maintaining that those paths which do not intersect any of the used
  separators still intersect a lot among themselves. 
  By balancedness, this will shrink the number of paths quickly, but by high intersection, there will always
  be many paths left, giving a contradiction.

  Define $q\df \lceil \log_{4\over 3} {\bigl({{|\Pp|}\over 24k}\bigr)}\rceil $. 
  We inductively define two sequences of linkages $\Pp=\Pp_0\speq \Pp_1\speq\cdots\speq\Pp_q$ and $\Qq=\Qq_0\speq \Qq_1\speq\cdots\speq\Qq_q$ and prove that they satisfy the following conditions for each $i \in \{0,\ldots,q\}$.
  \begin{itemize}
  \item[(i)] If $i > 0$, then $|\Pp_i|\le{3\over4}|\Pp_{i-1}|$.
  \item[(ii)] There exist quarter-integral linkages
    $\Pp_i^\backlink, \Qq_i^\backlink$ which are dual to $\Pp_i$ and
    $\Qq_i$, respectively.
  \item[(iii)] The minimum degree of $I(\Pp_i,\Qq_i)$ is at least $d-8ik$.
  \end{itemize}

  For the induction beginning, we define $\Pp_0 \df \Pp$ and
  $\Qq_0 \df \Qq$. By the preconditions of the lemma, it is clear that
  the above conditions are satisfied; for Condition~(iii), observe
  that $\Pp^\backlink$ and $\Qq^\backlink$ represent the required dual
  linkages $\Pp^\backlink_0$ and $\Qq^\backlink_0$.

  Now suppose that $i > 0$ and that $\Pp_{i - 1}$ and $\Qq_{i - 1}$
  have already been defined and that they satisfy the conditions. Let
  $A_i$ be the starting set of linkage $\Pp_{i-1}$, that is,
  $A_i = \linkfrom{\Pp_{i - 1}}$. We use
  Lemma~\ref{lem:balanced-separation} with $W = A_i$ to get a separation $(X_i, Y_i)$
  and a corresponding separator $S_i \df X_i \cap Y_i$ of size at most
  $k$ such that $X_i$ and $Y_i$ both contain at least $|A_i|/4$ elements of $A_i$. To see that Lemma~\ref{lem:balanced-separation} is applicable, recall that $d \geq 8kq + 24k + 4$ and thus \[|A_i| = |\Pp_{i-1}| \geq d-8k(i-1) \geq 8kq + 24k + 4 - 8k(i - 1) \geq 2k+2.\]
  Recall that there is no directed path from $Y_i$ to $X_i$ avoiding $S_i$.
  We define
  \[
    \Pp_i \df \{P\in\Pp_{i-1}\mid P\cap X_i=\emptyset\}\text{\qquad and \qquad} \Qq_i\df \{Q\in\Qq_{i-1}\mid Q\cap X_i=\emptyset\}.
  \]
  Clearly, we have $\Pp_i \subseteq \Pp_{i - 1}$ and $\Qq_i \subseteq \Qq_{i - 1}$. We claim that Conditions~(i) to~(iii) are satisfied. Condition~(i)
  is straightforward since at least $1\over 4$ of the paths $\Pp_i$ start in $X_i$. 

  Now consider Condition~(ii). 
  We define $\Pp_i^\backlink$ to be the backlinkage induced by $\Pp_i$ on $(\Pp,\Pp^\backlink)$ and
  $\Qq_i^\backlink$ to be a backlinkage induced by $\Qq_i$ on $(\Qq,\Qq^\backlink)$.
  Since $\Pp$, $\Pp^\backlink$, $\Qq$, and $\Qq^\backlink$ are half-integral, $\Pp_i^\backlink$ and $\Qq_i^\backlink$ are quarter-integral.

It remains to show Condition~(iii). The condition is trivial if $i = 0$. If $i > 0$, we first prove the following claim:

\begin{claim}\label{cla:cutSize}
  At most $8k$ paths from linkage $\Dd \in \{\Pp_{i - 1},\Qq_{i - 1}\}$ with corresponding dual linkage $\Dd^\backlink$ can intersect both $Y_i$ and $X_i$.
\end{claim}

\begin{claimproof}
  Clearly, there are at most $2k$ paths where a vertex in $Y_i$ precedes a vertex in $X_i$ since such a path has to pass through~$S_i$. Say that such a path is of the \emph{first type}.
  In fact, there are at most $2k$ paths of the first type in the half-integral linkage $\Dd$.
 
  Next, we bound the number of paths $P \in \Dd$ that go from a vertex in $X_i$ to a vertex in $Y_i$ and are not of the first type; say that such paths $P$ are of the \emph{second type}.
  We claim that there is an injective mapping~$M$, mapping each path~$P$ of the second type to some path $Q \in \Dd \cup \Dd^\backlink$ such that $Q$ has nonempty intersection with~$S_i$. 
  First, observe that $P$ has to start in $X_i$, because otherwise it is also of the first type.
  Denote by $s \df \pathfrom{P} \in X_i$ the starting vertex of $P$.
  Since $\Dd^\backlink$ is dual to $\Dd$, there is a path $Q_1 \in \Dd^\backlink$ that ends in~$s$.
  Either $Q_1$ intersects $S_i$, in which case we put $M(P) \df Q_1$, or not.
  In the second case, there is a path $Q_2 \in \Dd$ with $\pathto{Q_2} = \pathfrom{Q_1}$.
  Again, either $Q_2$ intersects $S_i$, in which case we put $M(P) \df Q_2$, or not.
  Continuing in this way, we will find~$Q_i \in \Dd \cup \Dd^\backlink$ such that $Q_i$ intersects~$S_i$ since, in each step in which $Q_i$ does not intersect~$Y_i$ the number of paths in $(\Dd \cup \Dd^\backlink) \setminus \{Q_i \mid i \in \mathbb{N}\}$ decreases, and there is at least one path in $(\Dd \cup \Dd^\backlink) \setminus \{Q_i \mid i \in \mathbb{N}\}$ which does intersect~$Y_i$; namely the path $R \in \Dd^\backlink$ with $\pathto{P} = \pathfrom{R}$. Furthermore, by definition no path in $\Dd \cup \Dd^\backlink$ will be defined as~$Q_i$ for two different paths~$P$. Thus, the mapping~$M$ that we construct is injective.

  Let $\mathcal{R}$ be the set of paths of the second type. Observe that
  $|M(\mathcal{R}) \cap \Dd^\backlink| \leq 4k$ since $\Dd^\backlink$ is
  quarter-integral by Condition~(iii). Furthermore,
  $|M(\mathcal{R}) \cap \Dd| \leq 2k$ since $\Dd$ is half-integral.

  Thus, overall there are at most $8k$ paths in $\Dd$ that intersect
  both $X_i$ and $Y_i$.
\end{claimproof}

Now we can prove Condition~(iii) when $i > 0$. We first show that there is at least one path $P$ in $\Pp_i$.
Let $\Pp_{i - 1}^Y$ be the set of paths in~$\Pp_i$ that start in~$Y_i$.
Note that $\Pp_i \subseteq \Pp_{i - 1}^Y$.
By choice of the separation~$(X_i, Y_i)$, we have $|\Pp_{i - 1}^Y| \geq |\Pp_{i - 1}|/4$.
By Condition~(iii) of the induction assumption we have $|\Pp_{i - 1}| \geq d - 8(i - 1)k$ and thus $|\Pp_{i - 1}^Y| \geq (d - 8(i - 1)k)/4$.
Since each path in $\Pp_{i - 1}^Y$ intersects~$Y_i$, Claim~\ref{cla:cutSize} shows that at most $8k$ paths in~$\Pp_{i - 1}^Y$ intersect~$X_i$.
Thus, the number of paths in $\Pp_i$ is at least $|\Pp_{i - 1}^Y| - 8k \geq (d - 8(i - 1)k)/4 - 8k$. Since $d\ge 8kq+24k+4$ by precondition of the lemma, we have 
\[{1\over 4} (d-8k(i-1)) - 8k \ge {1\over 4} (d-8ki+8k -32k) \ge {1\over 4} (d-8ki-24k)
  \ge 1.\]
Thus, indeed, there is a path $P \in \Pp_i$.
Path $P$ intersects with at least $d-8k(i-1)$ paths in $\Qq_{i-1}$ by the induction assumption.
At most $8k$ of them intersect with $X_i$ so $|\Qq_i|\ge d-8ki$.
This gives us several paths in $\Qq_i$ avoiding $X_i$.
We apply the previous argument symmetrically on one such path in $\Qq_i$ to get $|\Pp_i| \ge d-8ki$.
To conclude the proof of Condition (iii) observe that such arguments hold in fact for each path in either $\Pp_i, \Qq_i$.

We finish the proof of the lemma by showing that Conditions (i) and (iii) are in contradiction for some~$i \in [q]$.
Observe that these two conditions imply $d - 8ki \le |\Pp_i| \le ({3\over 4})^i|\Pp_0|$.
We show that $d - 8kq > ({3\over 4})^q|\Pp_0|$.
Since the conditions hold for $i = 0$, there is thus some smallest $i \in [q]$ for which $\Pp_i$ and $\Qq_i$ are well defined but the Conditions (i) and (iii) contradict each other.
Since $d>8kq+24k+4$ by precondition of the lemma, we have $d - 8kq > 24k + 4$. By definition of $q$ on the other hand \[({3\over 4})^q|\Pp_0| = \frac{|\Pp|}{{4 \over 3}^{\lceil \log_{4 \over 3}(|\Pp|/24k)\rceil}} \leq \frac{|\Pp|}{{4 \over 3}^{\log_{4 \over 3}(|\Pp|/24k)}} = 24k.\] Thus, indeed $d - 8kq > ({3\over 4})^q|\Pp_0|$, giving the desired contradiction.
\end{proof}

\paragraph{Main proof of the dense case}
We are now ready to prove the main lemma of this section.
\begin{proof}[Proof of Lemma~\ref{lem:dense}]
  Let $d = 327\,680\cdot a \cdot b \cdot \log_2(|\Ll|/b)$. Since 
  $I(\Ll, \Kk)$ is not $d$-degenerate, it contains an induced
  subgraph~$I'$ of minimum degree larger than~$d$. Redefine $\Ll$ and
  $\Kk$ to be the sublinkages of $\Ll$ and $\Kk$ contained in this
  subgraph~$I'$, that is, $\Ll \df \Ll \cap V(I')$ and
  $\Kk := \Kk \cap V(I')$. Note that $|\Ll| > d$, $|\Kk| > d$, the
  size of $\Ll$ only decreases, that is, it remains true that $d \geq 327\,680\cdot a \cdot b \cdot \log_2(|\Ll|/b)$, and note that
  $\linkfrom{\Ll} \cup \linkto{\Ll} \cup \linkfrom{\Kk} \cup
  \linkto{\Kk}$ remains well-linked.

  Let $\Ll^\backlink$ be
  a linkage in~$D$ from $\linkto{\Ll}$ to $\linkfrom{\Ll}$ and let
  $\Kk^\backlink$ be a linkage in~$D$ from $\linkto{\Kk}$ to
  $\linkfrom{\Kk}$. Note that $\Ll^\backlink$ and $\Kk^\backlink$
  exist because
  $\linkfrom{\Ll} \cup \linkto{\Ll} \cup \linkfrom{\Kk} \cup
  \linkto{\Kk}$ is well linked.

  We focus on $\auxg(\Ll)$ and $\auxg(\Kk)$.
  Take backlinkage-induced orderings $(L_1, \ldots, L_{|\Ll|})$
  of $\Ll$ 
  and $(K_1, \ldots, K_{|\Kk|})$ of $\Kk$.
  Apply Lemma~\ref{lem:disjointPairs} with $k = a$, $r = 640b\log_2(|\Ll|/b)$, 
  $G = I(\Ll, \Kk)$, $X = \{L_1, \ldots, L_{|\Ll|}\}$, and
  $Y = \{K_1, \ldots, K_{|\Kk|}\}$, obtaining $a$ sets
  $U_1, \ldots, U_a$ and $a$ sets
  $W_1, \ldots, W_a$ with the corresponding properties. To see that Lemma~\ref{lem:disjointPairs} is applicable, observe that $I(\Ll, \Kk)$ has minimum degree at least $327\,680 \cdot a \cdot b \log_2(|\Ll|/b) = 2^9 \cdot 640b\log_2(|\Ll|/b) \cdot a = 2^9 \cdot r \cdot k$.
  Observe for later
  on that, for each $i \in [a]$, the intersection graph $I(U_i, W_i)$ of the two linkages $U_i$ and $W_i$ has average degree at least $640b\log_2(|\Ll|/b)$ by property 3\ of Lemma~\ref{lem:disjointPairs}.

  Now define,
  for each $i \in [a]$, a graph $D_i$ as follows. Initially, take the
  union of all paths in $U_i$ and $W_i$. Then, for each edge~$(L, L')$
  of $\auxg(\Ll)$ such that $L, L' \in U_i$, add to $D_i$ the unique
  path $P \in \Ll^\backlink$ that connects $L$ and $L'$, that is,
  $\pathto{L} = \pathfrom{P}$ and $\pathto{P} = \pathfrom{L'}$.
  Similarly, for each edge~$(K, K')$ of $\auxg(\Kk)$ such that
  $K, K' \in W_i$, add to $D_i$ the unique path $Q \in \Kk^\backlink$ with
  $\pathto{K} = \pathfrom{Q}$ and $\pathto{Q} = \pathfrom{K'}$. In
  formulas:
  \begin{multline*}
    U'_i \ \df\ \{P \in \Ll^\backlink \mid \exists (L, L') \in E(\auxg(\Ll))
    \colon \\
    L, L' \in U_i \wedge \pathto{L} = \pathfrom{P}\wedge
    \pathto{P} = \pathfrom{L'}\}
  \end{multline*}  
  and
 \begin{multline*}
    W'_i \ \df\  \{Q \in \Kk^\backlink \mid \exists (K, K') \in E(\auxg(\Kk))
    \colon \\
    K, K' \in W_i \wedge \pathto{K} = \pathfrom{Q}\wedge
    \pathto{Q} = \pathfrom{K'}\}.
  \end{multline*}  
  We set
  \[ D_i \ \df\ \bigcup(U_i \cup W_i \cup U'_i \cup W'_i).\]
  
  We claim that $D_i$ satisfies the required properties. Clearly,
  $D_i$ is a subgraph of~$D$, giving property~(i). To see
  property~(ii), consider a
  linkage~$\Pp \in \{\Ll, \Ll^\backlink, \Kk, \Kk^\backlink\}$. We claim that no two
  subgraphs $D_i$, $D_j$ contain the same path of~$\Pp$. This claim
  follows indeed from property 2.\ of Lemma~\ref{lem:disjointPairs},
  stating that $U_i \cap U_j = \emptyset$ and
  $W_i \cap W_j = \emptyset$ and inspecting the definition of~$D_i$
  and~$D_j$. Thus, $\{V(D_i) \mid i \in [a]\}$ is a partition of a
  subset of the vertex set~$V(\Pp)$ of the paths in~$\Pp$. Thus, each
  vertex $v \in V(D)$ occurs in at most four subgraphs~$D_i$, showing
  property~(ii).

  It remains to show property~(iii), the lower bound on the directed
  treewidth of~$D_i$. We aim to modify~$D_i$, increasing the directed
  treewidth by at most a constant, to obtain a graph~$D^{(2)}_i$ which
  is the union of two pairs of dual half-integral linkages such that two linkages contained in distinct pairs intersect a lot. Then we can apply
  Lemma~\ref{lem:high-degree-linkages-dtw}, giving a lower bound
  on the directed treewidth of~$D^{(2)}_i$ which then implies a lower
  bound on the directed treewidth of~$D_i$.

  \newcommand{\zeroout}{\textsf{out}}
  \newcommand{\zeroin}{\textsf{in}} We first modify~$D_i$ to obtain a
  graph~$D^{(1)}_i$ which is the union of two pairs of dual linkages.
  Recall the orderings $\vec{\Ll} \df (L_1, \ldots, L_{|\Ll|})$ and
  $\vec{\Kk} \df (K_1, \ldots, K_{|\Kk|})$ on $\Ll$ and $\Kk$,
  respectively, which we have defined above. By property~1.\ of
  Lemma~\ref{lem:disjointPairs}, $U_i$ is a segment of~$\vec{\Ll}$ and
  $W_i$ is a segment of~$\vec{\Kk}$. Hence, by the way we have defined
  $\vec{\Ll}$, there are at most two
  cycles~$C$ in $\auxg(\Ll)$ which are not contained in $U_i$ or disjoint with $U_i$, that is
  $V(C) \setminus U_i \neq \emptyset$ and $V(C) \cap U_i \neq \emptyset$. Call such a cycle
  \emph{broken}. Similarly, there are at most two cycles $C$ in
  $\auxg(\Kk)$ such that $V(C) \setminus W_i \neq \emptyset$ and $V(C) \cap W_i \neq \emptyset$. 
  Call
  such a cycle \emph{broken} as well. For each broken cycle~$C$, do
  the following operation on~$D_i$ to obtain~$D^{(1)}_i$. If $C$ is in
  $\auxg(\Ll)$, let $L^C_\zeroout$ be the vertex of outdegree zero in
  the subgraph $\auxg(\Ll)[V(C) \cap U_i]$ and let $L^C_\zeroin$
  be the vertex of indegree zero. Add the directed
  edge~$(\pathto{L^C_\zeroout}, \pathfrom{L^C_\zeroin})$ to~$D_i$.
  Proceed analogously if $C$ is in $\auxg(\Kk)$: Let $K^C_\zeroout$ be
  the vertex of outdegree zero in the subgraph $\auxg(\Kk)[V(C) \cap W_i]$
  and let $K^C_\zeroin$ be the vertex of indegree zero,
  and add the directed
  edge~$(\pathto{K^C_\zeroout}, \pathfrom{K^C_\zeroin})$ to~$D_i$. In
  this way, we add at most four edges to~$D_i$, obtaining~$D^{(1)}_i$.
  Note that adding an edge increases the directed treewidth by at most
  one\footnote{In the corresponding robber-cop game (see~\cite{JohnsonRST01}), we can always guard the new edge with
    an additional cop.}, and hence
  $\dtw(D^{(1)}_i) \leq \dtw(D_i) + 4$.

  We claim that $D^{(1)}_i$ is the union of two pairs of dual
  linkages. To see this, note first that $U_i$ and $W_i$ are linkages
  in~$D^{(1)}_i$. Now consider
  \[U^b_i \ \df\ U'_i \cup \{(\pathto{L^C_\zeroout}, \pathfrom{L^C_\zeroin}) \mid C
    \text{ a broken cycle in }\auxg(\Ll)\}\]
  and
  \[W^b_i \df W'_i \cup \{(\pathto{K^C_\zeroout},
    \pathfrom{K^C_\zeroin}) \mid C \text{ a broken cycle in
    }\auxg(\Kk)\},\] wherein $L^C_\zeroin, L^C_\zeroout, K^C_\zeroin$,
  and $K^C_\zeroout$ are defined as above. Clearly,
  $D^{(1)}_i = \bigcup (U_i \cup W_i \cup U^b_i \cup W^b_i)$.
  Moreover, both $U^b_i$ and $W^b_i$ are linkages because $U'_i$ and
  $W'_i$ are linkages and because
  $L^C_\zeroin, L^C_\zeroout, K^C_\zeroin$, and $K^C_\zeroout$ have
  indegree or outdegree zero in $\auxg(\Ll)[V(C)]$ or
  $\auxg(\Kk)[V(C)]$, respectively. Finally, by definition, $U_i$ and
  $U^b_i$ are dual to each other and $W_i$ and $W^b_i$ are dual to
  each other. Thus, $D^{(1)}_i$ is the union of two pairs of dual
  linkages, as claimed.
  
  In order to apply Lemma~\ref{lem:high-degree-linkages-dtw}, we need
  a pair of linkages whose intersection graph has a large minimum
  degree. So far, the linkages which define~$D^{(1)}_i$ guarantee only
  large average degree (via property 3.\ of
  Lemma~\ref{lem:disjointPairs}). We now derive a subgraph $D^{(2)}_i$ of $D^{(1)}_i$ such that $D^{(2)}_i$ is the union of
  two pairs of dual half-integral linkages
  $(\Pp, \Pp^\backlink), (\Qq, \Qq^\backlink)$ and $I(\Pp, \Qq)$ has large minimum
  degree.
  To achieve this, recall
  that the intersection graph $I(U_i, W_i)$ of the two linkages $U_i$,
  $W_i$ in $D^{(1)}_i$ has average degree at least $640b\log_2(|\Ll|/b)$. 
  Hence, there is a subgraph~$I'$ of $I(U_i, W_i)$ with minimum degree
  at least~$320b\log_2(|\Ll|/b)$. 
  Let $\Pp \subseteq U_i$ be the sublinkage of $U_i$
  contained in $I'$, that is $\Pp = U_i \cap V(I')$. Similarly, let
  $\Qq = W_i \cap V(I')$.

  We define $\Pp^\backlink$ to be the backlinkage induced by $\Pp$
  on $(U_i, U^b_i)$ and $\Qq^\backlink$ to be the backlinkage induced
  by $\Qq$ on $(W_i,W^b_i)$. Note that $\Pp^\backlink$ and $\Qq^\backlink$
  are half-integral and dual to $\Pp$ and $\Qq$, respectively.

  Take now the subgraph~$D^{(2)}_i$ to be the union
  $\bigcup(\Pp \cup \Pp^\backlink \cup \Qq \cup \Qq^\backlink)$. Then
  apply Lemma~\ref{lem:high-degree-linkages-dtw} to
  $\Pp, \Pp^\backlink, \Qq, \Qq^\backlink$ with $k = b + 4$ and
  $d = 320b\log_2(|\Ll|/b)$. To see that the preconditions of 
  Lemma~\ref{lem:high-degree-linkages-dtw} are satisfied, first recall
  that the intersection graph $I(\Pp, \Qq)$ has minimum degree at
  least $320b\log_2(|\Ll|/b)$. Furthermore, 
  \begin{multline*}
    d = 320b\log_2\frac{|\Ll|}{b} \geq 200b \log_2\frac{|\Ll|}{b} + 120b + 4 \geq \frac{5 \cdot 40b}{2} \log_2\frac{|\Ll|}{b} + 120b + 4 \geq {}\\
    \frac{8 \cdot 5b}{\log_2(4/3)}\log_2\frac{|\Ll|}{b} + 24(5b) + 4 \geq 8
    \cdot (b + 4) \log_{4/3}\frac{|\Ll|}{24(b + 4)} + 24(b + 4) + 4 =\\
     8k\log_{4/3}\frac{|\Ll|}{24k} + 24k + 4,
  \end{multline*}
  and thus indeed the preconditions of
  Lemma~\ref{lem:high-degree-linkages-dtw} are satisfied. Thus, the
  directed treewidth of $D^{(2)}_i$ is at least $b + 4$. Since
  $D^{(2)}_i$ is a subgraph of $D^{(1)}_i$ and
  $\dtw(D_i) \geq \dtw(D^{(1)}_i) - 4$, we have $\dtw(D_i) \geq b$, as
  required.
\end{proof}

\section{Wrapping up the proof of Theorem~\ref{thm:qp}}\label{sec:main}

\begin{proof}[Proof of Theorem~\ref{thm:qp}]
  Let $G$ be a directed graph of $\dtw(G) \geq c \cdot a^6b^{8}\log^2(ab)$, where $c$ is a large constant, whose value will follow from the reasoning below. First, we invoke Lemma~\ref{lem:path-system} with $\beta=2^{37}a^2b^3\log(ab)$ and $\alpha=8ab$ (here we assume that $c$ is sufficiently large so that the assumption is satisfied).   We obtain a set of vertex-disjoint paths $P_1,\ldots,P_{8ab}$ and sets $A_i,B_i\subseteq V(P_i)$, where $A_i$ appears before $B_i$ on $P_i$, and $|A_i|= |B_i|= 2^{37}a^2b^3\log(ab)$, and the set
$
  \bigcup_{i=1}^{8ab} A_i\cup B_i
$
is well-linked.
Denote by $\Ll_{i,j}$ a linkage from $B_i$ to $A_j$. 

We split the $8ab$ paths $P_i$ into $a$ segments,
each consisting of $8b$ paths. 
Formally, for every $\iota \in [a]$ we define $I_\iota = \{j~|~8(\iota-1)b < j \leq 8\iota b\}$.

Now we set $r = 64ab^2$ and  create an auxiliary $r$-colored graph $H$, whose vertices will be paths of appropriately chosen linkages $\Ll_{i,j}$. More specifically, for every $\iota \in [a]$, and every $i,j \in I_\iota$, we introduce a vertex for every path in $\Ll_{i,j}$ and color it $(i,j)$. Two vertices of $H$ are adjacent if and only if their corresponding paths share a vertex in $G$. Note that for two linkages $\Ll_{i,j}$ and $\Ll_{i',j'}$, the graph $H[\Ll_{i,j} \cup \Ll_{i',j'}]$ is precisely the intersection graph $I(\Ll_{i,j},\Ll_{i',j'})$.

We set $d\df2^{27} ab\log(ab)$ and consider two cases:
\begin{enumerate}
  \item[(i)] for all $i,j,i',j'$ the graph $I(\Ll_{i,j},\Ll_{i',j'})$ is $d$-degenerate. 
  \item[(ii)] there exist $i,j,i',j'$, for which the graph $I(\Ll_{i,j},\Ll_{i',j'})$ is not $d$-degenerate.
\end{enumerate}
An intuition behind case (i) is that for each subgraph of $H$ there is always a path (in $G$) such that it shares a vertex with at most $d$ paths from all used linkages back.

\smallskip
{\bf Case (i)}~~We use Lemma~\ref{lem:degenerate} on $H$.
Graph $H$ has $64ab^2$ color classes such that for each $(i,j)\neq (i',j')$ the graph $H[\Ll_{i,j}\cup\Ll_{i',j'}]$ is $d$-degenerate.
Note that  $|\Ll_{i,j}|=2^{37}a^2b^3\log(ab)\ge 4e(r-1)d$ is sufficiently large to satisfy the last assumption of the lemma.
We are given an independent set $x_1,\ldots, x_{64ab^2}$ that represents pairwise disjoint paths $L_{i,j}$ from $B_i$ to $A_j$ for all $i,j \in I_\iota$. 
We also recall that $A_i$ and $B_i$ lie on $P_i$ and all $P_i$'s are pairwise disjoint.

Let $G_\iota$ consist of all paths $P_i$ for $i \in I_\iota$ and $L_{i,j}$ for $i,j \in I_\iota$.
  By Lemma~\ref{lem:dtw-bound} for $k=8b$ we obtain $\dtw(G_\iota)\ge b$ while each vertex is in at most 2 such subgraphs.
  Indeed, each vertex can appear only once on some $P_i$ and once on some $L_{i,j}$.

\smallskip

{\bf Case (ii)}~~The claim follows from Lemma~\ref{lem:dense}.
Since $|\Ll|=2^{37}a^2b^3\log(ab)$ then
$d=2^{27} ab\log(ab)>2^{19}ab\log(2^{37}a^2b^2\log(ab))$.
\end{proof}

\section{Improved bound for cycles: Proof of Theorem~\ref{thm:dtw-ep-all}}\label{sec:imp}

This section is devoted to the proof of Theorem~\ref{thm:dtw-ep-all}.
We follow the outline of Section~\ref{sec:main}, but circumvent the usage of Lemma~\ref{lem:path-system} to avoid the quadratic blow-up stemming from it.

The proof of the cases $p=4$, $p=3$, and $p=2$ differ only in minor details. We first present the proof for the case $p=4$ in Section~\ref{ss:imp1},
    abstracting the common parts of the proofs as independent lemmas, and then continue with the proof of the case $p=2$ in Section~\ref{ss:imp2}.
A simple mixture of the tricks used for the cases $p=4$ and $p=2$ yields the proof for the case $p=3$ and is discussed in Section~\ref{ss:imp3}.

\subsection{Case $p=4$}\label{ss:imp1}
The crucial replacement of Lemma~\ref{lem:path-system} is the following.
\begin{lemma}\label{lem:walk-system}
Let $G$ be a directed graph, $a, b, k \geq 1$ be integers, and let $D$ be a well-linked set in $G$ of size $4(a+k)b$. 
If $G$ does not contain a family of $k$ cycles such that every vertex of $G$ is in at most two of the cycles, then 
there exists a family $\Pp = \{P_1,P_2,\ldots,P_a\}$ of walks in $G$ and sets $A_i, B_i \subseteq V(P_i)$ for every $1 \leq i \leq a$ such that
\begin{enumerate}
\item $\Pp$ is of congestion $2$,
\item the sets $A_i$ and $B_i$ are of size $b$ each and are pairwise disjoint, 
\item for every $1 \leq i \leq a$, all vertices of $A_i$ appear on $P_i$ before all vertices of $B_i$, and
\item $\bigcup_{i=1}^a A_i \cup B_i$ is well-linked in $G$.
\end{enumerate}
\end{lemma}
Lemma~\ref{lem:walk-system} differs from Lemma~\ref{lem:path-system} in a number of ways. 
First, it avoids the quadratic blow-up in the size of the well-linked set (which is linearly lower bounded by directed treewidth by Lemma~\ref{lem:dtw2wl}). 
Second, $\Pp$ is no longer a linkage but a family of walks of congestion $2$. 
Third, there is another assumption that $G$ does not contain a large half-integral packing of cycles; 
we do not know how to avoid this assumption and this assumption is the reason the improvement described here works only in the setting of Theorem~\ref{thm:dtw-ep-all},
not in the general setting of Theorem~\ref{thm:qp}.

\begin{proof}[Proof of Lemma~\ref{lem:walk-system}.]
Partition $D$ into two sets $D_1$ and $D_2$ of size $2(a+k)b$ each.
By well-linkedness, there exists a linkage $\Ll$ from $D_1$ to $D_2$ and a linkage $\Ll^\backlink$ from $D_2$ to $D_1$.
We focus on the auxiliary graph $\auxg(\Ll)$ and a backlinkage-induced
order $\Ll = \{L_1,L_2,\ldots,L_{|\Ll|}\}$. 
Note that $\auxg(\Ll)$ has less than $k$ connected components, since the closed
walks in $G$ corresponding to the cycles of $\auxg(\Ll)$ give rise to a half-integral
packing of cycles in $G$.
We say that an index~$i \in \{1, 2, \ldots, a+k\}$ is \emph{good} if all vertices $L_j$ for $(i-1)\cdot 2b < j \leq i \cdot 2b$ lie on the same cycle of $\auxg(\Ll)$, and \emph{bad} otherwise.
Note that we have less than $k$ bad indices. Let $I$ be a family of exactly $a$ good indices.

For every $i \in I$, we define $P_i$ to be the walk in $G$ that corresponds to the path $\{ L_j~|~(i-1) \cdot 2b < j \leq i \cdot 2b \}$ in $\auxg(\Ll)$.
Furthermore, let $A_i = \{\pathfrom{L_j}~|~(i-1) \cdot 2b < j \leq i \cdot 2b - b\}$ and $B_i = \{\pathfrom{L_j}~|~i \cdot 2b - b < j \leq i \cdot 2b\}$. 
Then clearly $\Pp = \{P_i~|~i \in I\}$ is of congestion $2$; the other required properties are straightforward to verify.
\end{proof}

With Lemma~\ref{lem:walk-system} in hand, we can closely follow the reasoning of Section~\ref{sec:main}.
We first formulate and prove two lemmas which we will reuse in the next section.
We start with the sparse scenario.

\begin{lemma}\label{lem:sparse-win}
Let $a,b,d$ be positive integers with $a$ even and $b \geq 4e \cdot a\cdot d$, and $G$ be a directed graph.
Let $\Pp=P_1,\ldots,P_{a}$ be a set of paths of congestion $\alpha$ such that there exist pairwise disjoint sets $A_i,B_i \subseteq V(P_i)$, $i=1,2,\ldots,a$.
Furthermore, assume that each set $A_i$ and $B_i$ is of size $b\ge 4ead$ and that
for every $1 \leq i \leq a$, all vertices of $A_i$ appear on $P_i$ before all vertices of $B_i$. 
Let $\mathcal{I} = \{(1,2), (2,1), (3,4), (4,3), \ldots, (a-1,a), (a,a-1)\}$. 
For every $(i,j) \in \mathcal{I}$, let $\mathcal{L}_{i,j}$ be a linkage from $B_i$ to $A_j$. 

If for every $(i,j),(i',j') \in \mathcal{I}$, $(i,j) \neq (i',j')$, the intersection graph $I(\Ll_{i,j}, \Ll_{i',j'})$ is $d$-degenerate,
then there exist a family of $a\over 2$ directed cycles of congestion $\alpha +1$.
\end{lemma}

\begin{proof}
%
Create an auxiliary $a$-partite graph $H$ with vertex sets of color classes equal to $\Ll_{i,j}$ for $(i,j) \in \mathcal{I}$.
Between $\Ll_{i,j}$ and $\Ll_{i',j'}$ put the graph $I(\Ll_{i,j},\Ll_{i',j'})$.
By Lemma~\ref{lem:degenerate} and our choice of $b$, there
exists $L_{i,j} \in \Ll_{i,j}$ for every $(i,j) \in \mathcal{I}$
that are independent in $H$.
By the construction of the graph $H$, the paths $L_{i,j}$ for $(i,j) \in \mathcal{I}$ are
pairwise vertex-disjoint.

Fix $\iota \in \{1, 2, \ldots, {a\over 2}\}$ and consider the union $U_\iota$ of $P_{2\iota-1}$, $P_{2\iota}$,
$L_{2\iota-1,2\iota}$, and $L_{2\iota,2\iota-1}$. Observe
that this union contains a closed walk: from the ending vertex of $L_{2\iota,2\iota-1}$
follow $P_{2\iota-1}$ to the starting vertex of $L_{2\iota-1,2\iota}$, then 
follow $L_{2\iota-1,2\iota}$ to the end, then follow $P_{2\iota}$ to the starting
vetex of $L_{2\iota,2\iota-1}$, and follow this path to the end.
Thus, $U_\iota$ contains a cycle $C_\iota$. 
Furthermore, since every vertex can appear at most $\alpha$ times on walks $P_i$ and at most once on paths $L_{i,j}$, every vertex can appear at most $\alpha +1$ times on the cycles $\{ C_\iota~|~1 \leq \iota \leq {a\over 2}\}$.
\end{proof}

For the core of the complementary (dense) situation, we derive the following lemma.
Consider backlinkage-induced order $\Ll = \{L_1,L_2,\ldots,L_{|\Ll|}\}$ for linkage $\Ll$ and the corresponding backlinkage $\Ll^\backlink$.
We say that a walk (path) is an \emph{($\Ll$,$\Ll^\backlink$)-interlaced walk (path) of size $q\ge 1$} if it starts at $\pathfrom{L_j}$ for some $L_j\in \Ll$ and then it has the following structure:
    \[
      L_{j},L_{j}^\backlink, L_{j+1},L_{j+1}^\backlink,\ldots,L_{j+q-1}.
    \]
We may omit the size when it only matters whether such a walk exists.
    
\begin{lemma}\label{lem:dense-win}
  Let $\Ll$ and $\Kk$ be two linkages in a directed graph $G$.
  Let $U_1,\ldots,U_k$ be a set of $k$ walks such that the congestion of $(U_i)_{i=1}^k$ is $\alpha$ and $\mathcal{U}_i$ is the family of paths of $\Ll$ that are subpaths of $U_i$.
  Similarly, let $W_1,\ldots,W_k$ be a set of $k$ walks such that the congestion of $(W_i)_{i=1}^k$ is $\beta$ and $\mathcal{W}_i$ is the family of paths of $\Kk$ that are subpaths of $W_i$.

If for every $1 \le i \leq k$ the average degree of $I(\Ll,\Kk)[\mathcal{U}_i,\mathcal{W}_i]$ is at least 2, then there exists in $G$ a family of $k$ cycles of congestion $\alpha +\beta$.
\end{lemma}

\begin{proof}
Fix $i \in \{1,\ldots,k\}$.
Let $L_1,L_2,\ldots$ be the paths of $\mathcal{U}_i$ in the order of their appearance on $U_i$
and let $K_1,K_2,\ldots$ be the paths of $\mathcal{W}_i$ in the order of their appearance on $W_i$.
Since the average degree of $I(\Ll,\Kk)[\mathcal{U}_i,\mathcal{W}_i]$ is at least $2$, this graph is not a forest.
Consequently, there are indices
$\alpha < \beta$ and
$\gamma < \delta$ such that
$L_\alpha K_\delta \in E(I(\Ll,\Kk)[\mathcal{U}_i,\mathcal{W}_i])$ and $L_\beta K_\gamma \in E(I(\Ll,\Kk)[\mathcal{U}_i,\mathcal{W}_i])$. 
Consider the following closed walk $Q_i$ in $G$: 
starting from the intersection of $L_\alpha$ and $K_\delta$, we follow $U_i$
up to the intersection with $K_\gamma$. Then we follow $W_i$
up to the intersection with $L_\alpha$, where we started the walk. 
Let $Q_i'$ be any cycle inside the closed walk $Q_i$.
Thus we obtained $k$ cycles.
Observe that as we build the cycles only using vertices in $\cup_{i=1}^k U_i\cup W_i$, every vertex of $G$ is used at most $\alpha+\beta$ times.
\end{proof}

\medskip

We conclude the proof of case $p=4$ in Theorem~\ref{thm:dtw-ep-all} by a combination of Lemmas~\ref{lem:walk-system},~\ref{lem:sparse-win}, and~\ref{lem:dense-win}.
Let $k$ be an integer and $G$ be a directed graph of $\dtw(G) = \Omega(k^3)$ and suppose, for a contradiction, that no family of $k$ cycles exists such that every vertex of $G$ is in at most four of the cycles.
Let
\begin{align*}
d & := 2^{10} \cdot k, & a & := 2k, & b & := \lceil 4ead \rceil = \Theta(k^2). 
\end{align*}
By Lemma~\ref{lem:dtw2wl}, $G$ contains a well-linked set of size $\Omega(k^3)$. 
We apply Lemma~\ref{lem:walk-system} to $G$ with parameters $a$ and $b$, obtaining
a family $\Pp = \{P_1,P_2,\ldots,P_a\}$ and sets $A_i$, $B_i$ of size $b$ each.

Let $\mathcal{I} = \{(1,2), (2,1), (3,4), (4,3), \ldots, (a-1,a), (a,a-1)\}$. 
For every $(i,j) \in \mathcal{I}$, let $\mathcal{L}_{i,j}$ be a linkage from $B_i$ to $A_j$.
We consider two cases.

In the case where for every $(i,j),(i',j') \in \mathcal{I}$, $(i,j) \neq (i',j')$, the intersection graph $I(\Ll_{i,j}, \Ll_{i',j'})$ is $d$-degenerate we get a contradiction by Lemma~\ref{lem:sparse-win}.
For the remaining case observe that 
there exist a linkage $\Ll \subseteq \Ll_{i,j}$ and a linkage
$\Kk \subseteq \Ll_{i',j'}$ such that $I(\Ll,\Kk)$ has minimum degree more than $d$.
Furthermore, since $\bigcup_{i=1}^a A_i \cup B_i$ is well-linked, there exist a linkage $\Ll^\backlink$ from $B(\Ll)$ to $A(\Ll)$ 
 and an analogous linkage $\Kk^\backlink$ from $B(\Kk)$ to $A(\Kk)$.

We focus on auxiliary graph $\auxg(\Ll)$ and $\auxg(\Kk)$.
Let $\Ll = \{L_1,L_2,\ldots,L_{|\Ll|}\}$ and $\Kk = \{K_1,K_2,\ldots,K_{|\Kk|}\}$
be backlinkage-induced orders of $\Ll$ and $\Kk$.
Let $L_j^\backlink$ be the path of $\Ll^\backlink$ that starts at $\pathto{L_j}$
and similarly define $K_j^\backlink$.
Since $G$ does not admit a quarter-integral packing of cycles of size $k$, we infer that both $\auxg(\Ll)$ and $\auxg(\Kk)$ have each less than $k$ connected components.

We now apply Lemma~\ref{lem:disjointPairs} to $I(\Ll,\Kk)$ with the aforementioned backlinkage-induced
orders of $\Ll$ and $\Kk$, aiming at $3k$ sets
$U_1,\ldots,U_{3k}$ and $3k$ sets $W_1,\ldots,W_{3k}$ such that
$I(\Ll,\Kk)[U_i,W_i]$ has average degree at least $2$
for every $1 \leq i \leq 3k$.

An index $1 \leq i \leq 3k$ is bad if either $U_i$ is not contained in a single cycle of $\auxg(\Ll)$
or $W_i$ is not contained in a single cycle of $\auxg(\Kk)$. 
By our orderings of $\Ll$ and $\Kk$, there are less than $2k$ bad indices. 
Let $I \subseteq [3k]$ be a family of exactly $k$ indices that are not bad.
We can now use Lemma~\ref{lem:dense-win}.
For each $i\in I$, $U_i$ can be turned into ($\Ll$,$\Ll^\backlink$)-interlaced walk $U_i'$.
Similarly each $W_i$ can be turned into ($\Kk$,$\Kk^\backlink$)-interlaced walk $W_i'$.
 $k$ blow-upThe congestion of $(U'_i)_{i \in I}$ is two as it is composed of two linkages, and similarly for $(W'_i)_{i \in I}$. 
Therefore we obtain a quarter-integral cycle packing of size $k$, a contradiction.
This finishes the proof of case $p=4$ in Theorem~\ref{thm:dtw-ep-all}.

\subsection{Case $p=2$}\label{ss:imp2}

First, we prove a lemma that serves as a key technique to lower the congestion.

\begin{lemma}[Untangling Lemma]\label{lem:untangle}
  Let $G$ be a directed graph, let $q, k \geq 1$ be integers, and let $D_1,D_2$ be two vertex sets of size $q(2k-2)+1$ each.
Let $\Ll$ be linkage from $D_1$ to $D_2$ of size $q(2k-2)+1$ in graph $G$ and $\Ll^\backlink$ be a corresponding back-linkage of size $q(2k-2)+1$, too. 
If $G$ does not admit a half-integral packing of $k$ cycles, then 
$G$ contains an ($\Ll$,$\Ll^\backlink$)-interlaced path of size $q$.
\end{lemma}

\begin{proof}
We iteratively define subgraphs $H_1,H_2,\ldots$ of $G$ using the following greedy process.
Let $\Ll = \{L_1,L_2,\ldots \}$ be the backlinkage-induced order of $\Ll$.
Fix $i \geq 1$ and assume that all $H_{i'}$ for $1 \leq i' < i$ have been defined.
Let $j$ be smallest index such that $L_j$ was not used for construction of $H_{i'}$ for any $i'$ with $1\le i'<i$.
The subgraph $H_i$ is defined as the following walk.
Starting in $\pathfrom{L_j}$, we follow $L_j$, the path of $\Ll^\backlink$ from $\pathto{L_j}$ to $\pathfrom{L_{j+1}}$, 
$L_{j+1}$, etc.,  until we reach either an end of a cycle of $\auxg(\Ll)$ or a self-intersection of the walk. 
In the latter case, let $H_i$ be the walk from $\pathfrom{L_j}$ up to and including the last arc leading to the self-intersection.
We measure the \emph{size of $H_i$} as the number of vertices paths $L_{j'}$ for which we passed $\pathfrom{L_{j'}}$ in the construction.

Now, we observe that as $\mathcal{H} := \{H_1,H_2,\ldots\}$ is created using $\Ll$ and $\Ll^\backlink$ only, so it has congestion $2$.
Furthermore, every $H_i$ whose greedy process ended because of a self-intersection contains a cycle. 
Since $G$ does not contain a half-integral packing of $k$ cycles, $\auxg(\Ll)$ has less than $k$ cycles and thus for less than $k$ walks $H_i$ the greedy process ended because
of a self-intersection. Consequently, $|\mathcal H| \leq 2k-2$.
Hence, there exists $H_i\in \mathcal H$ of size at least $q+1$.
It follows that $H_i$ contains the desired ($\Ll$,$\Ll^\backlink$)-interlaced path of size $q$. 
\end{proof}

Second, we give an analog of Lemma~\ref{lem:walk-system} that serves as a replacement of Lemma~\ref{lem:path-system} in this section.
This time, we trade linear blow-up in the exponent for no congestion.

\begin{lemma}\label{lem:disj-walk-system}
  Let $G$ be a directed graph, let $a, b, k \geq 1$ be integers, and let $D$ be a well-linked set in $G$ of size $2(ab(2k-2)+1)$. 
If $G$ does not contain a family of $k$ cycles such that every vertex of $G$ is in at most two of the cycles, then 
there exists a family $\Pp = \{P_1,P_2,\ldots,P_a\}$ of paths in $G$ and sets $A_i, B_i \subseteq V(P_i)$ for every $1 \leq i \leq a$ such that
\begin{enumerate}
\item the paths in $\Pp$ are mutually disjoint,
\item the sets $A_i$ and $B_i$ are of size $b$ each and are pairwise disjoint, 
\item for every $1 \leq i \leq a$, all vertices of $A_i$ appear on $P_i$ before all vertices of $B_i$, and
\item $\bigcup_{i=1}^a A_i \cup B_i$ is well-linked in $G$.
\end{enumerate}
\end{lemma}

\begin{proof}
We partition $D$ into two equal sets $D_1$ and $D_2$ of size $ab(2k-2)+1$ each.
By well-linkedness, there exists a linkage $\Ll$ from $D_1$ to $D_2$ and a backlinkage $\Ll^\backlink$ from $D_2$ to $D_1$.
This gives us the backlinkage-induced order $\Ll = \{L_1,L_2,\ldots,L_{|\Ll|}\}$. 
We immediately use Lemma~\ref{lem:untangle} with $q=ab$.

As $G$ does not contain a half-integral packing of $k$ cycles, we obtain an $(\Ll,\Ll^\backlink)$-interlaced path $P$
that contains at least $2ab$ vertices in $D=\{\pathfrom{L_j} \mid L_j\in\Ll\}\cup\{\pathto{L_j} \mid L_j\in\Ll\}$.
For every $j \in \{1, 2, \ldots, a\}$, we define $P_j$ to be $j$-th subpath of $P$ containing exactly $2b$ consecutive vertices from set $D$; the define $A_i$ to be the set of the first $b$
of these vertices and $B_i$ to be the set of the last $b$ of these vertices.
Then it is straightforward to verify that $\Pp$ satisfy the required properties.
\end{proof}

\medskip

We conclude the proof of case $p=2$ in Theorem~\ref{thm:dtw-ep-all} by combination of Lemmas~\ref{lem:untangle},~\ref{lem:disj-walk-system},~\ref{lem:sparse-win}, and~\ref{lem:dense-win}.

Let $k$ be an integer and $G$ be a directed graph of $\dtw(G) = \Omega(k^5)$ and suppose, for a contradiction, that no family of $k$ cycles exists such that every vertex is in at most two of the cycles.
Let
\begin{align*}
  d & := 3 \cdot 2^{10}\cdot k, & a & := 2k, & q & := \lceil 4ead \rceil, & b & := 2(q(2k-2)+1) = \Theta(k^3). 
\end{align*}
By Lemma~\ref{lem:dtw2wl}, $G$ contains a well-linked set of size $\Omega(k^5)$. 
We apply Lemma~\ref{lem:disj-walk-system} to $G$ with parameters $a$ and $b$, obtaining
a family $\Pp = \{P_1,P_2,\ldots,P_a\}$ and sets $A_i$, $B_i$ of size $b$ each.

Let $\mathcal{I} = \{(1,2), (2,1), (3,4), (4,3), \ldots, (a-1,a), (a,a-1)\}$. 
For every $(i,j) \in \mathcal{I}$, let $\mathcal{L}_{i,j}$ be a linkage from $B_i$ to $A_j$ and $\mathcal{L}^\backlink_{i,j}$ is the corresponding linkage back (which exists due to well-linkedness of $\bigcup_{i=1}^a A_i \cup B_i$). 
Now we will untangle all such linkages using Lemma~\ref{lem:untangle}.
We apply Lemma~\ref{lem:untangle} on each $(i,j)\in \mathcal{I}$ separately with the parameter $q$,
obtaining an $(\Ll_{i,j},\Ll_{i,j}^\backlink)$\nobreakdash-interlaced path $Q_{i,j}$ containing at least $q$ vertices in $A_i$ and at least $q$ vertices in $B_j$.
Let $\mathcal{Q}_{i,j}$ be the sublinkage of $\Ll_{i,j}$ consisting of $q$ paths contained in $Q_{i,j}$.
We consider two cases.

In the case where, for every $(i,j),(i',j') \in \mathcal{I}$, $(i,j) \neq (i',j')$, the intersection graph $I(\mathcal{Q}_{i,j}, \mathcal{Q}_{i',j'})$ is $d$-degenerate we get a contradiction by Lemma~\ref{lem:sparse-win} as $\Pp$ has congestion one.
In the remaining case, fix two distinct $(i,j),(i',j') \in \mathcal{I}$ such that $I(\mathcal{Q}_{i,j},\mathcal{Q}_{i',j'})$ is not $d$-degenerate.
We have a linkage $\mathcal{Q}_1 \subseteq \mathcal{Q}_{i,j}$ and a linkage
$\mathcal{Q}_2 \subseteq \mathcal{Q}_{i',j'}$ such that $I(\mathcal{Q}_1,\mathcal{Q}_2)$ has minimum degree more than~$d$.
We now apply Lemma~\ref{lem:disjointPairs} to $I(\mathcal{Q}_1,\mathcal{Q}_2)$, aiming at $k$ sets $U_1,\ldots,U_k$ and $k$ sets $W_1,\ldots,W_k$ such that $I(\mathcal{Q}_1,\mathcal{Q}_2)[U_\iota,W_\iota]$ has average degree at least $2$ for every $\iota \in \{1, 2, \ldots, k\}$.
In the application of Lemma~\ref{lem:disjointPairs}, the paths in $\mathcal{Q}_1$ and $\mathcal{Q}_2$ are ordered as in $Q_{i,j}$ and $Q_{i',j'}$, respectively. 
Hence, all paths in $U_\iota$ are contained in a subpath $U_{\iota}'$ of $Q_{i,j}$ and the paths $(U_\iota')_{\iota = 1}^k$ are vertex-disjoint.
Similarly, 
all paths in $W_\iota$ are contained in a subpath $W_{\iota}'$ of $Q_{i',j'}$ and the paths $(W_\iota')_{\iota = 1}^k$ are vertex-disjoint.
We can now use Lemma~\ref{lem:dense-win} to get a contradiction. Thus case $p=2$ of Theorem~\ref{thm:dtw-ep-all} holds. 

\subsection{Case $p=3$}\label{ss:imp3}
We conclude with a remark that we can combine both approaches.
If we use Lemma~\ref{lem:walk-system} instead of Lemma~\ref{lem:disj-walk-system} in the proof of case $p=2$ of Theorem~\ref{thm:dtw-ep-all}, we are guaranteed only $k$ one-third-integral cycles (as $\Pp$ has congestion 2 instead of 1 while using Lemma~\ref{lem:sparse-win}) but we save blow-up by factor $k$ in the bound on directed treewidth.
Hence, we obtain the statement of Theorem~\ref{thm:dtw-ep-all} for $p=3$.

\section{Conclusions}
We have shown that if one relaxes the disjointness constraint to half- or quarter-integral packing (i.e., every vertex used at most two or four times, respectively), then
the Erd\H{o}s-P\'{o}sa property in directed graphs admits a polynomial bound between
the cycle packing number and the feedback vertex set number. 
The obtained bound for quarter-integral packing is smaller than the one for half-integral packing.
A natural question would be to decrease the dependency further, even at the cost of higher
congestion (but still a constant). 
More precisely, we pose the following question: Does there exist a constant $c$ and a polynomial $p$ such that
for every integer $k$ if a directed graph $G$ does not contain a family of $k$ cycles
such that every vertex of $G$ is in at most $c$ of the cycles, then 
the directed treewidth of $G$ is at most $k p(\log k)$?

One of the sources of polynomial blow-up in the proof of Theorem~\ref{thm:qp}
is the quadratic blow-up in Lemma~\ref{lem:path-system}. 
Lemma~\ref{lem:path-system} is a direct corollary of another result 
of~\cite{DBLP:journals/corr/KawarabayashiK14} that asserts that a directed graph $G$
of directed treewidth $\Omega(k^2)$ contains a path $P$ and a set $A \subseteq V(P)$
that is well-linked and of size $k$. 
Is this quadratic blow-up necessary? Can we improve it, even at the cost of some constant
congestion in the path $P$ (i.e., allow $P$ to visit every vertex a constant number of times)?
We remark that the essence of the improvement from $\Oh(k^6 \log^2 k)$ (obtained by setting $b=2$ in Theorem~\ref{thm:qp}) to $\Oh(k^3)$ asserted by Theorem~\ref{thm:dtw-ep-all} for $p=4$ is to avoid the usage
of Lemma~\ref{lem:path-system} and to replace it with a simple well-linkedness trick.
However, this trick fails in the general setting of Theorem~\ref{thm:qp}.

\paragraph{Acknowledgments.} We thank Stephan Kreutzer (TU Berlin) for
interesting discussions on the topic and for pointing out
Lemma~\ref{lem:dtw-cp-fvs}.

\marginpar{\includegraphics[height=25px]{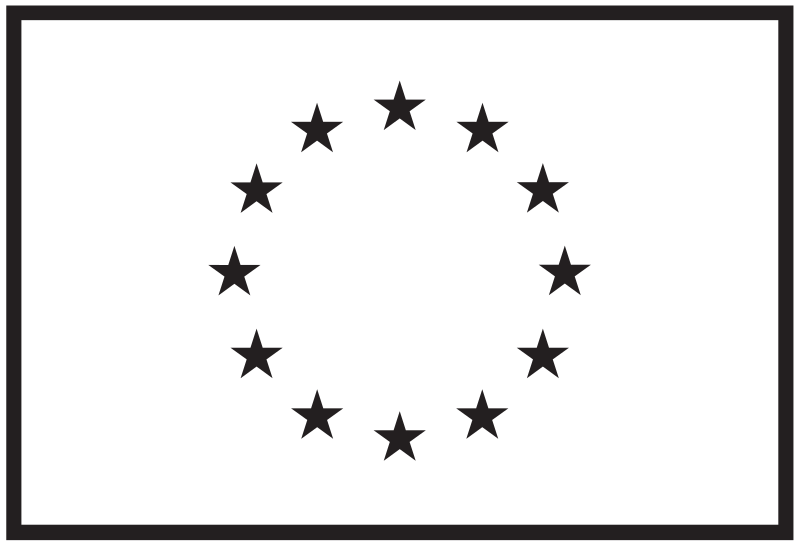}\hspace{.5cm}\includegraphics[height=25px]{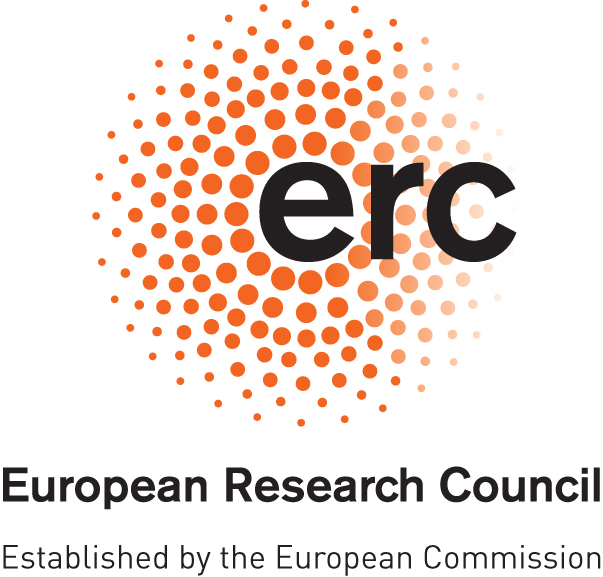}}

This research is part of projects that have received funding from the
European Research Council (ERC) under the European Union's Horizon
2020 research and innovation programme Grant Agreements 648527 (Irene
Muzi) and 714704 (all authors). 

Tom\'{a}\v{s} Masa\v{r}\'{i}k was also supported by student grant
number SVV–2017–260452 of Charles University, Prague, Czech Republic.

\bibliographystyle{abbrv}
\bibliography{main}

\end{document}